\documentclass[a4paper]{IEEEtran}

\pdfoutput=1

\usepackage[australian]{babel}
\usepackage[T1]{fontenc} 
\usepackage{lmodern}
\usepackage{amssymb}
\usepackage{amsmath}
\usepackage{amsthm}
\usepackage{mathrsfs} 
\usepackage{textcomp} 
\usepackage[numbers,square,compress]{natbib}

\usepackage{graphicx}
\usepackage[caption=false]{subfig}
\usepackage{nicefrac}
\usepackage{fixmath}
\usepackage{bm}
\usepackage{mdwmath} 
\usepackage{mdwtab}
\usepackage{fixltx2e} 
\usepackage{array}
\usepackage[printonlyused,withpage]{acronym} 
\usepackage{microtype} 

\newcommand{\mat}[1]{\mathbold{#1}} 
\newcommand{\matTxt}[1]{\boldsymbol{\mathit{#1}}}     

\makeatletter
\def\imod#1{\allowbreak\mkern10mu({\operator@font mod}\,#1)}
\makeatother

\DeclareMathAlphabet{\mathpzc}{OT1}{pzc}{m}{it} 

\DeclareMathSizes{10}{9}{7}{5}

\newcommand{\placement}{htb}

\newcommand{\primeSymbol}{\textit{p}}

\newcommand{\firstIndex}{j}

\newcommand{\FFTfirstIndex}{u}

\newcommand{\eqnTag}{Eq.}

\newcommand{\figTag}{Fig.}
\newcommand{\figsTag}{Figs}

\newcommand{\secTag}{Sec.}

\newcommand{\corTag}{Cor.}
\newcommand{\propTag}{Prop.}
\newcommand{\propsTag}{Props}

\newcommand{\numberSymbol}{N}
\newcommand{\ghostNumberSymbol}{\mathpzc{N}}
\newcommand{\ghostSymbol}{\mathcal{G}}
\newcommand{\modulusSymbol}{M}
\newcommand{\primRootSymbol}{a}
\newcommand{\rootSymbol}{\alpha}

\newcommand{\MinkowskiSymbol}{\oplus}

\newcommand{\FourierSymbol}{\mat{F}}

\newcommand{\VandermondeSymbol}{\mat{V}}
\newcommand{\CircVandermondeSymbol}{\mat{W}}
\newcommand{\inputVecSymbol}{\matTxt{x}}

\newcommand{\wholeNumbers}{\mathbb{N}_0}

\newcommand{\integerNumbers}{\mathbb{Z}}

\newtheorem{theorem}{Proposition}
\newtheorem{corollary}{Corollary}
\newtheorem{definition}{Definition}

\title{Recovering Missing Slices of the Discrete Fourier Transform using Ghosts}
\author{Shekhar Chandra, Imants Svalbe\thanks{Shekhar Chandra and Imants Svalbe are with the School of Physics, Monash University, Australia. Email: Shekhar.Chandra@monash.edu or Imants.Svalbe@monash.edu}, Jeanpierre Gu\'edon, Andrew Kingston\thanks{Andrew Kingston is with the Department of Applied Mathematics, Australian National University.} and Nicolas Normand\thanks{Jeanpierre Gu\'edon and Nicolas Normand are with IRCCyN-IVC, \'{E}cole polytechnique de l'Universit\'{e} de Nantes, France.}}

\begin{document}

\maketitle

 \begin{abstract}
 The \ac{DFT} underpins the solution to many inverse problems commonly possessing missing or un-measured frequency information. This incomplete coverage of Fourier space always produces systematic artefacts called Ghosts. In this paper, a fast and exact method for de-convolving cyclic artefacts caused by missing slices of the \ac{DFT} is presented. The slices discussed here originate from the exact partitioning of \ac{DFT} space, under the projective \acl{DRT}, called the Discrete \acl{FST}. The method has a computational complexity of $O(n\log_2 n)$ (where $n = N^2$) and is constructed from a new Finite Ghost theory. This theory is also shown to unify several aspects of work done on Ghosts over the past three decades. The paper concludes with a significant application to fast, exact, non-iterative image reconstruction from sets of discrete slices obtained for a limited range of projection angles.
 \end{abstract}

\acresetall

\section{Introduction}\label{sec::Intro}
The \ac{DFT} is an important tool for inverse problems, where the Fourier representation of an object is used as a mechanism to recover that object. For example, in \ac{CT}, the internal structure of an object is recovered or reconstructed from its projected ``views'' or projections~\citep{Kak2001}. The \ac{FT} of the projections can be placed into Fourier space and the inverse \ac{FT} is used to reconstruct the object. This is especially advantageous given the efficient and low computationally complex algorithm of the \ac{DFT} known as the \ac{FFT}~\citep{Cooley1965}. In such cases, the acquired data cannot fully cover Fourier space as the problem is ill-posed~\citep{Logan1975, Smith1977}. Thus, there are missing Fourier coefficients and it is common practice to interpolate the space from the known Fourier coefficients. The choice of interpolation method is a major factor in determining the quality of the reconstruction~\citep{Gottleib2000, Walden2000}.

\subsection{Ghosts}
Incomplete Fourier coverage leads to the introduction of reconstruction artefacts known as ``Ghosts'' or ``invisible distributions''~\citep{Bracewell1954}. Ghosts are effectively formed from an under-determined set of projections, i.e. from the ``missing'' or unmeasured projections. Therefore, Ghosts are always present within \ac{CT} reconstructions and are the main reason for the filtering and interpolation of projection data~\citep{Logan1975}. The initial work on Ghosts in continuous space was pioneered by Bracewell and Roberts~\citep{Bracewell1954}, Logan~\citep{Logan1975}, Katz~\citep{Katz} and Louis~\citep{Louis1981}. Recent work by Cand\`{e}s~\emph{et al.}~\citep{Candes2006} has also addressed this issue using more modern methods.

In contrast, projection sets in \ac{DT}, which have practically well defined reconstruction processes~\citep{Herman1999}, are often deliberately or naturally under-determined for various applications. These applications range from image encoding~\citep{Normand2010}, network transmission~\citep{NormandEtAl} to tomography~\citep{Chandra2008}. Discrete Ghosts were first proposed by Katz~\citep{Katz} as a way to describe zero-valued discrete projections taken at rational angles $\theta_{pq}$, i.e. $\theta_{pq} = \tan^{-1}(\nicefrac{q}{p})$ (or simply the vector $[q,p]$) where $p,q \in \integerNumbers$. A simple example of these discrete Ghosts is shown in \figTag~\ref{fig::Ghosts}(a).
\begin{figure}[\placement]
 \centering
 \subfloat[]{
 \includegraphics[width=0.18\textwidth]{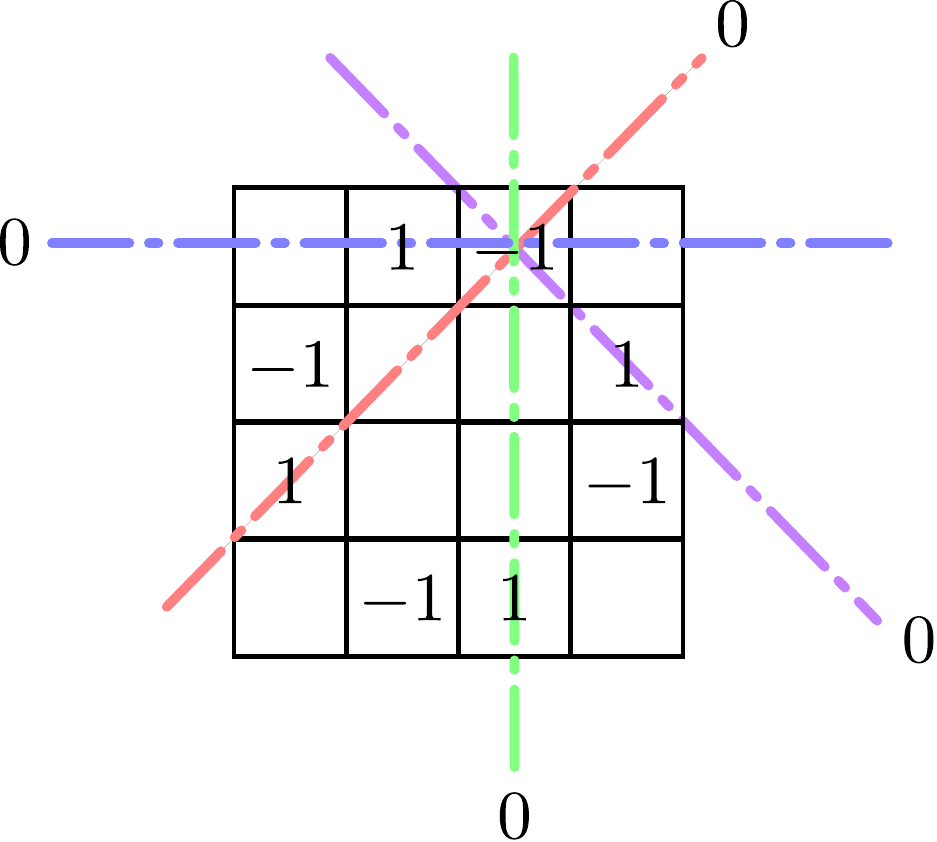}}
 \hspace{0.15cm}
 \subfloat[]{
 \includegraphics[width=0.19\textwidth]{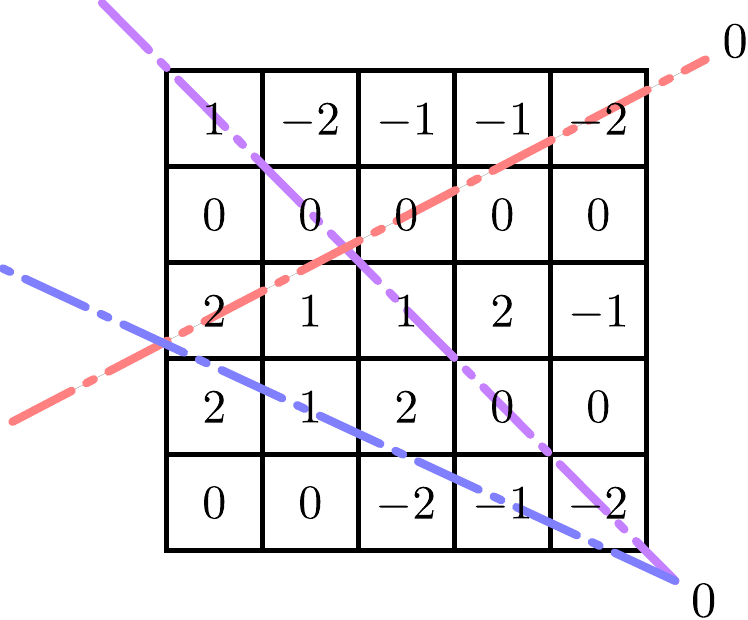}}
 \caption{An example of discrete and finite Ghosts. (a) shows a Ghost, as constructed by Katz~\citep{Katz}, which is invisible when projected at any of the four rational angles shown. (b) shows a finite (periodic) Ghost, which is invisible in three finite angles $(m=1,2,3)$ shown.}
 \label{fig::Ghosts}
\end{figure}
Katz~\citep{Katz} determined that an $\numberSymbol\times\numberSymbol$ image can be reconstructed exactly from a set of $\mu$ rational angle projections if and only if
\begin{eqnarray}
 \numberSymbol \leqslant 1 + \max\left(\sum_{\firstIndex=0}^{\mu-1} |p_\firstIndex| ,\ \sum_{\firstIndex=0}^{\mu-1} |q_\firstIndex| \right). \label{eqn::KatzCriterion}
\end{eqnarray}
This is now known as the Katz Criterion. It is a statement that the information contained in the projection set needs to be one-to-one with the image data. Thus, knowing whether the projections of a sub-region in the reconstruction meets the Katz Criterion will allow one to determine if the sub-region is exactly reconstructable.
\begin{figure}[\placement]
 \centering
 \subfloat[]{
 \includegraphics[width=0.21\textwidth]{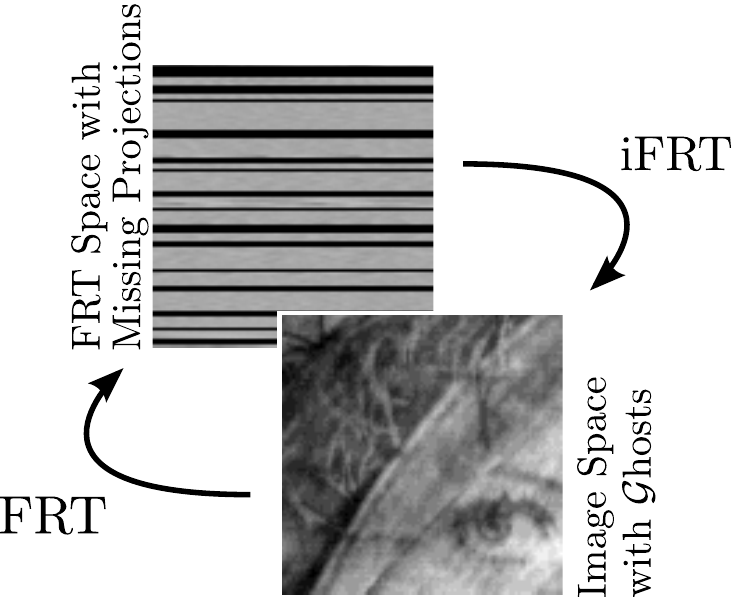}}
 \subfloat[]{
 \includegraphics[width=0.25\textwidth]{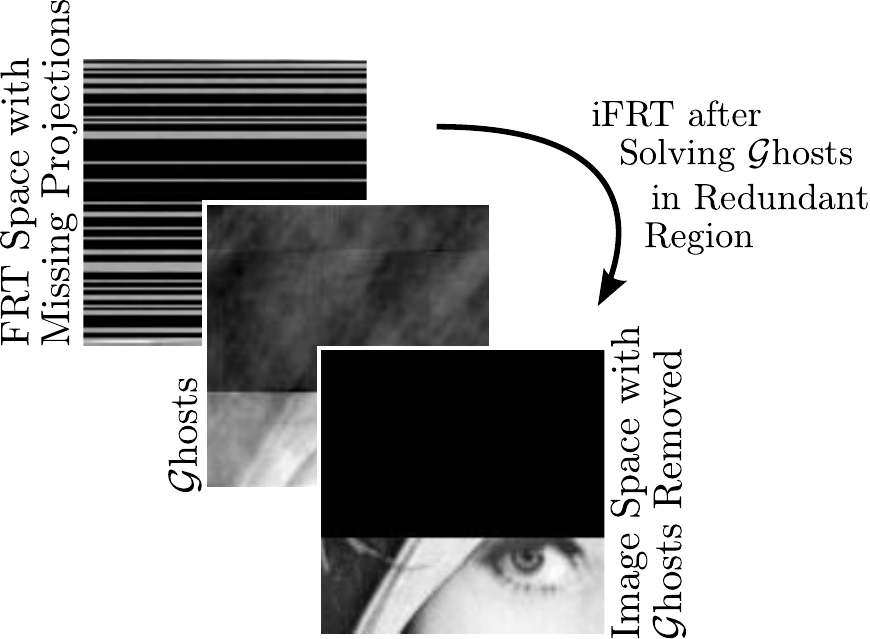}}
 \caption{An illustration of finite Ghosts formed within the \acl*{DRT}~(\acs*{DRT}). (a) shows missing projections (black rows in \acs*{DRT} space) and their effect on the reconstructed image. (b) shows the result of ``De-Ghosting'' in order to restore an image when a redundant image area is present.}
 \label{fig::FiniteGhosts}
\end{figure}

Chandra \emph{et al.}~\citep{Chandra2008, Chandra2008b} showed that Ghosts also exist in the \ac{DFT} as cyclic artefacts when \ac{DT} projection data is missing. These cyclic Ghosts will be referred to as ``Finite Ghosts'' (see \figTag~\ref{fig::Ghosts}(b)). They showed it was possible to recover missing projection data exactly in \ac{DT} given a sufficient redundant region in the reconstruction. However, the method could only solve for a very small number of missing projections while still being computationally viable. Their method was also based on empirical observations of the Finite Ghosts.
\begin{figure}[\placement]
 \centering
 \includegraphics[width=0.45\textwidth]{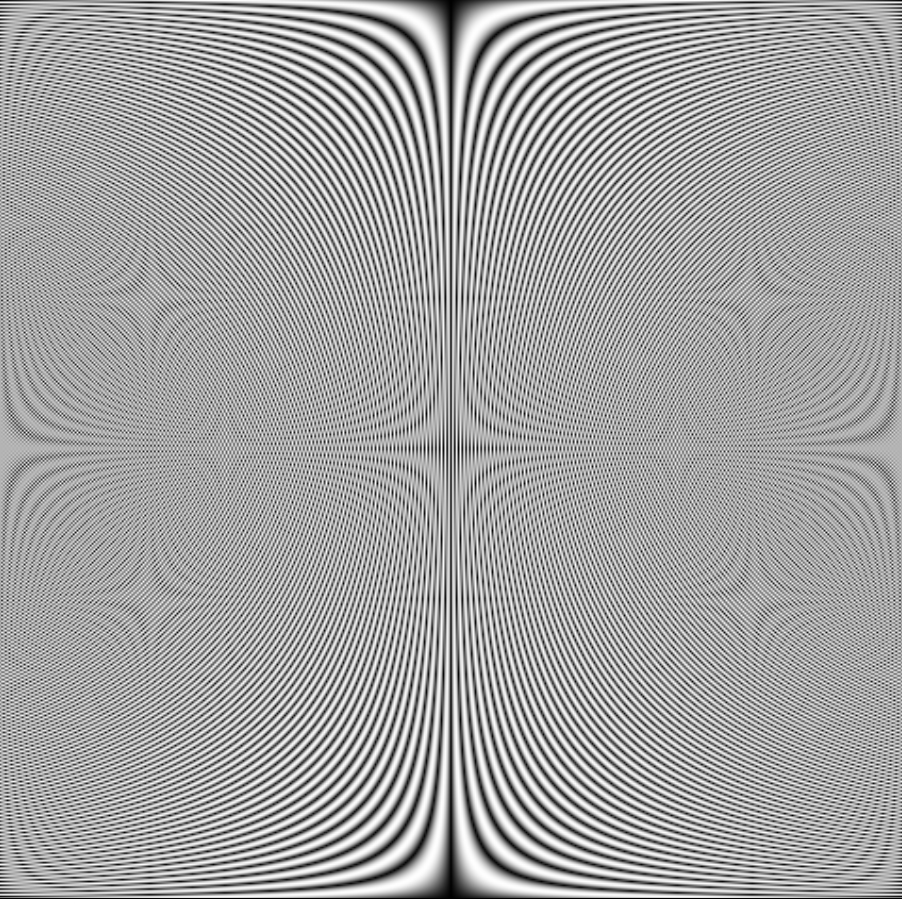}
 \caption{The Fourier eigenvalues of one dimensional circulant Ghost operators for $\numberSymbol=479$.}
 \label{fig::GhostOperators}
\end{figure}

In this paper, we present a fast method of recovering missing slices from \ac{DFT} space that may be applied efficiently to very large sets of missing projection data. The method allows exact recovery of these missing slices by removing Ghosts formed within the \ac{DFT}, a process that will be referred to as ``De-Ghosting''. The slices are equivalent to the projections of the \ac{DRT}. A schematic of this process is shown in \figTag~\ref{fig::FiniteGhosts} using the \ac{DRT}. The method will utilise a new theory of Ghosts in the \ac{DFT} constructed in \secTag~\ref{sec::FiniteGhosts}. A consequence of the theory are the special \ac{1D} Ghost operators shown in \figTag~\ref{fig::GhostOperators} from which the De-Ghosting is computed. The end result is a method to reconstruct an object exactly from a limited number of projections (with incomplete \ac{DFT} coverage) without interpolation (see \secTag~\ref{sec::Mojette}). The incomplete \ac{DFT} coverage in this case is restricted to the set of discrete slices of the \ac{DFT} as described in the next section. It is hoped that a deep understanding of such methods will allow more efficient approaches to inverse problems that utilise the \ac{DFT}.

\subsection{Discrete Fourier Slice Theorem}\label{sec::FST}
The \acf{DRT} provides an exact partitioning of \ac{DFT} space in the form of finite or cyclic slices~\citep[TIP]{Grigoryan2003}. An example of the slice partitioning is given in \figTag~\ref{fig::FST}.
\begin{figure}[\placement]
 \newcommand{\samplespacewidth}{0.19\textwidth}
 \newcommand{\samplehspace}{0.05cm}

 \centering
 \subfloat[]{
 \includegraphics[width=\samplespacewidth]{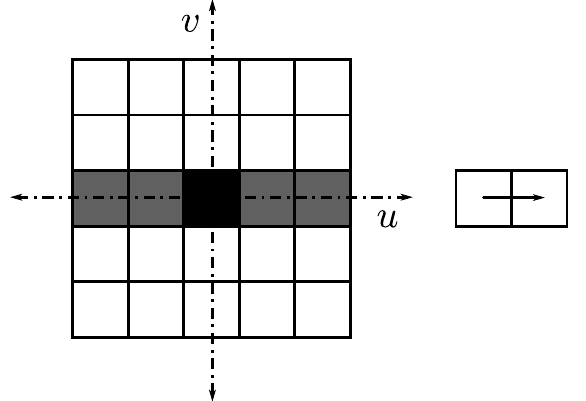}}
 \hspace{\samplehspace}
 \subfloat[]{
 \includegraphics[width=\samplespacewidth]{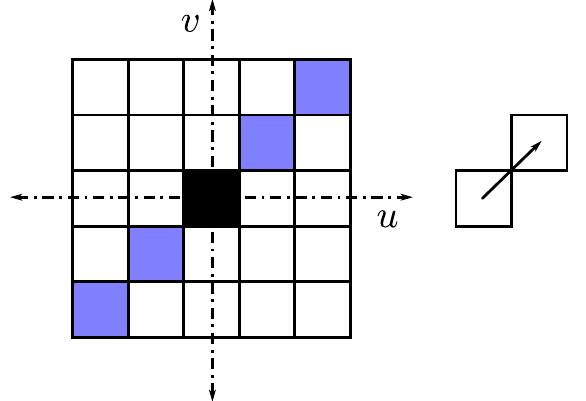}}
 
 \subfloat[]{
 \includegraphics[width=\samplespacewidth]{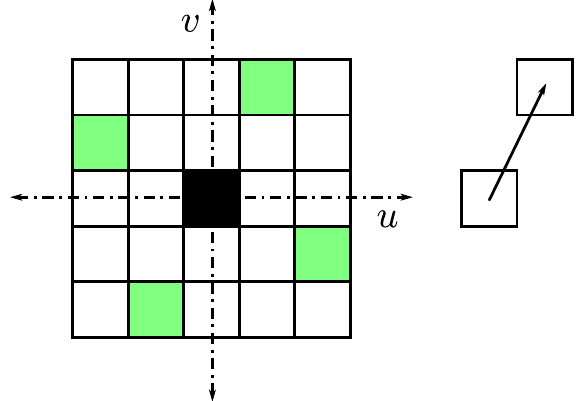}}
 \hspace{\samplehspace}
 \subfloat[]{
 \includegraphics[width=\samplespacewidth]{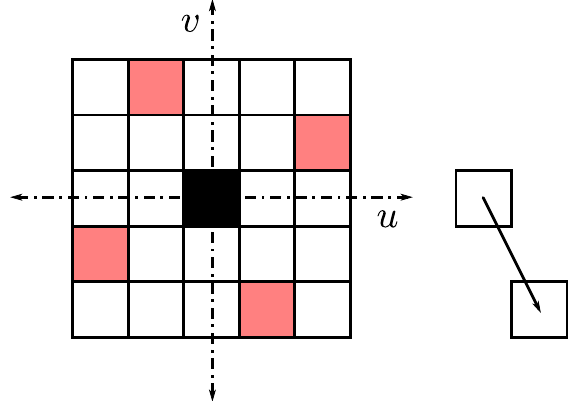}}

 \subfloat[]{
 \includegraphics[width=\samplespacewidth]{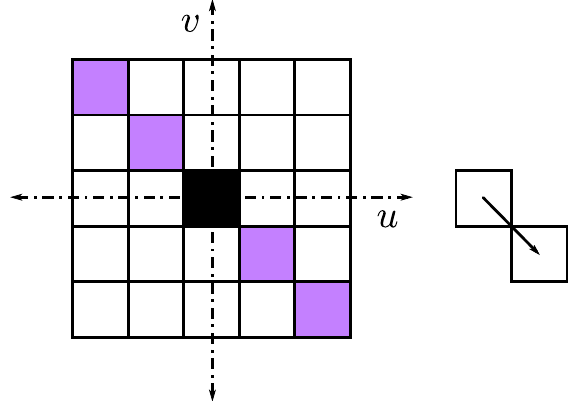}}
 \hspace{\samplehspace}
 \subfloat[]{
 \includegraphics[width=\samplespacewidth]{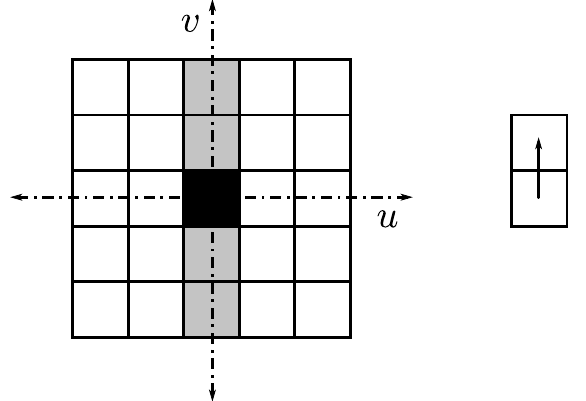}}

 \caption{The slices within the geometry of the \acf*{DFT} for a $5\times 5$ image. Each colour represents a slice of a different slope with the DC coefficient centred (black). Note that each vector shown is computed modulo $N$. (a)-(e) shows the slices with slopes $0 \leqslant m < 4 \imod 5$ in \acs*{DFT} space. (f) shows the row sum (perpendicular) slice in \acs*{DFT} space.}
 \label{fig::FST}
\end{figure}
The \ac{DRT} projections have the same geometry as the slices, i.e. they are taken as lines or congruences of the form
\begin{equation}
 y \equiv mx + t \imod N, \label{eqn::DiscreteLine}
\end{equation}
for an $\numberSymbol\times \numberSymbol$ image where $m$ is the slope, $t$ is the translate and $m, t, y, x \in \wholeNumbers$. The lines in this geometry wrap around the image in both rows and columns. 

In this paper, the term ``image'' is used interchangeably with a \ac{2D} ``object'' and assume that $\numberSymbol$ is prime, which defines the simplest form of the \ac{DRT}\footnote{See Hsung \emph{et al.}~\citep{Hsung1996} or Chandra \emph{et al.}~\citep{Chandra2010a} for a discussion of the dyadic \ac{DRT}.}. In general, one requires a total of $\numberSymbol + \nicefrac{\numberSymbol}{\primeSymbol}$ projections for an exact reconstruction for an $\numberSymbol\times \numberSymbol$ image, where $N=\primeSymbol^n$ (which includes powers of two). Thus, the projections are taken along the vectors $[1,m]$ (with $0 \leqslant m < N$) and $[0,1]$, and the slices are placed along the vectors $[-m,1]$ and $[1,0]$ in \ac{DFT} space (see examples shown as arrows in \figTag~\ref{fig::FST}). 

Chandra \emph{et al.}~\citep{Chandra2009, Chandra2010c} showed that these projections (and lines) are equivalent to circulant matrices (or simply circulants) for all translates $t$ (see \figTag~\ref{fig::Lines}).
\begin{definition}[Circulant~\citep{Davis1979}]
A circulant is an $N\times N$ matrix containing a unique row $f(\firstIndex)$ with $\firstIndex = 0,\ldots,N-1$ replicated on each row, but where each row is cyclically shifted $\!\!\imod N$ by a certain number of elements to the right.
\end{definition}
\begin{figure}[\placement]
 \centering
 \subfloat[]{
 \includegraphics[width=0.16\textwidth]{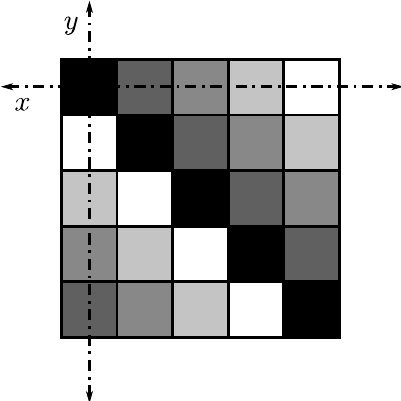}}
 \hspace{0.65cm}
 \subfloat[]{
 \includegraphics[width=0.18\textwidth]{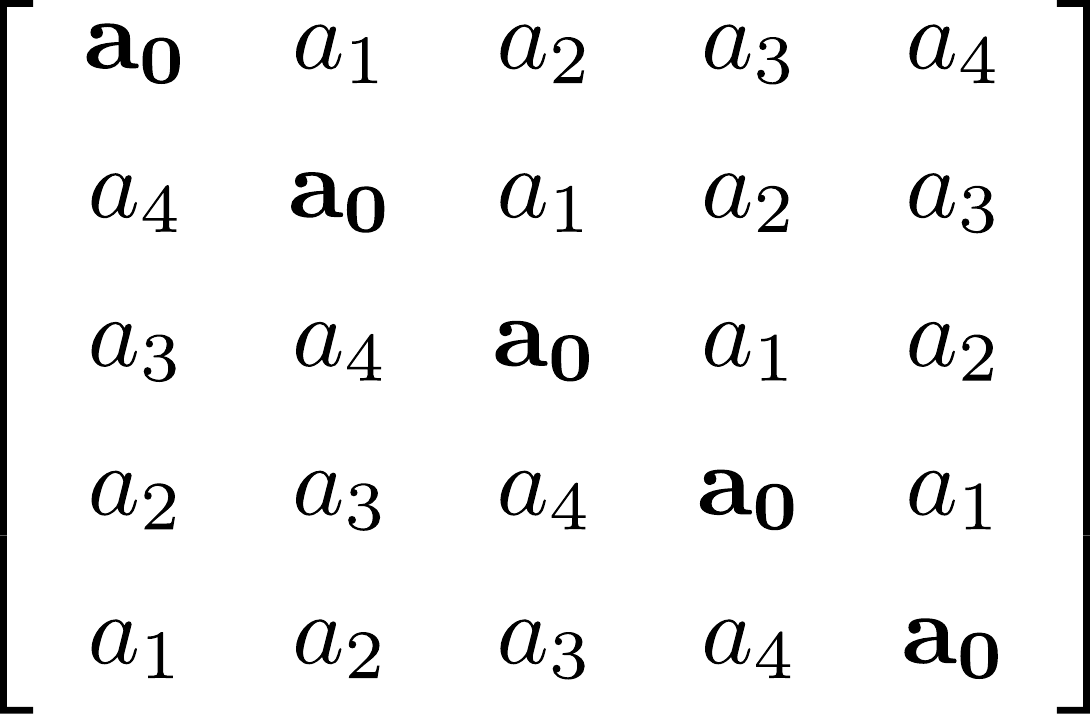}}
 \caption{The lines within the geometry of the \acl*{DFT} with slope 1 of a $5\times 5$ image. (a) shows the wrapping behaviour of the lines in an image lattice. Each grey scale shows pixels of the same translate. (b) shows the equivalent circulant matrix for the lines.}
 \label{fig::Lines}
\end{figure}
The property of circulants important here is that a circulant is diagonalised by the \ac{2D} \ac{DFT}. \eqnTag~\eqref{eqn::Diagonalise2} shows an example of this property, where the values $\lambda_\firstIndex$ are the eigenvalues of the circulant. 
\begin{eqnarray}
	\left[
	\begin{array}{ccccc}
	\bm{a_0} & a_1 & a_2 & a_3 & a_4\\
	a_3 & a_4 & \bm{a_0} & a_1 & a_2\\
	a_1 & a_2 & a_3 & a_4 & \bm{a_0}\\
	a_4 & \bm{a_0} & a_1 & a_2 & a_3\\
	a_2 & a_3 & a_4 & \bm{a_0} & a_1
	\end{array}
	\right] \quad\overset{\textrm{DFT/iDFT}}{\Longleftrightarrow}\quad
	\left[
	\begin{array}{ccccc}
	\lambda_0 &  &  &  &  \\ 
	 &  & \lambda_1 &  &  \\ 
	 &  &  &  & \lambda_2 \\ 
	 & \lambda_3 &  &  &  \\ 
	 &  &  & \lambda_4 & 
	\end{array}
	\right]\label{eqn::Diagonalise2}
\end{eqnarray}
Since each slice in the \ac{DFT} is obtained by diagonalising its corresponding circulant using the \ac{DFT} (see \eqnTag~\eqref{eqn::Diagonalise2}), the placing of each slice in \ac{DFT} space adds its corresponding circulant to image space until the image is recovered. This is known as \ac{CBP}~\citep{Chandra2010c} and an example of \ac{CBP} is given in \figTag~\ref{fig::CBP}. When all projections are diagonalised, the resulting slices fully tile all of \ac{DFT} space precisely once (see \figTag~\ref{fig::FST}) allowing the exact reconstruction of the object.
\begin{figure}[\placement]
 \centering
 \newcommand{\samplespacewidth}{0.12}
 \newcommand{\samplehspace}{0.25cm}

 \subfloat[1 Circulant]{
  \includegraphics[width=\samplespacewidth\textwidth]{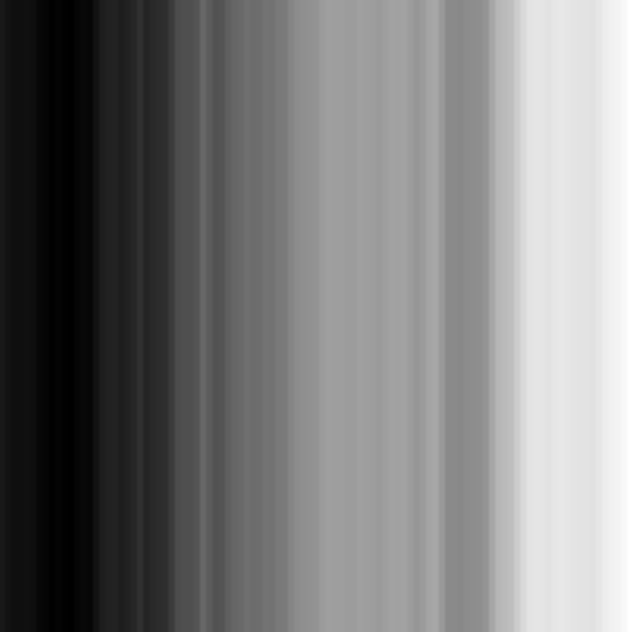}
 }
 \hspace{\samplehspace}
 \subfloat[20 Circulants]{
  \includegraphics[width=\samplespacewidth\textwidth]{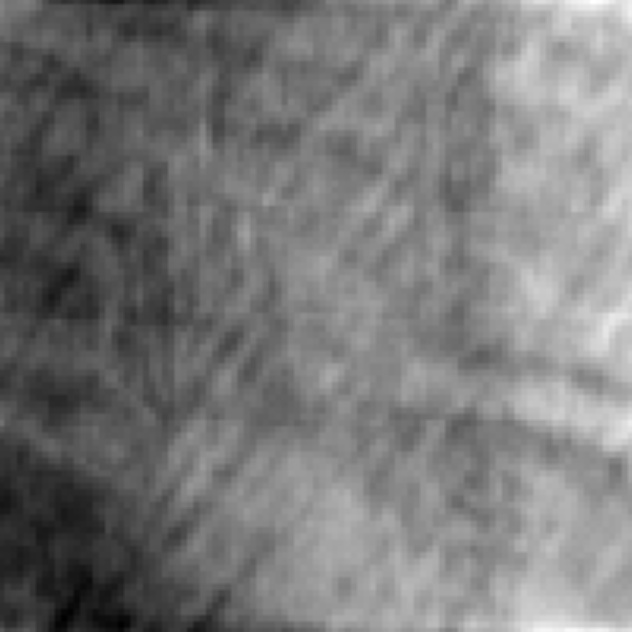}
 }
 \hspace{\samplehspace}
 \subfloat[40 Circulants]{
  \includegraphics[width=\samplespacewidth\textwidth]{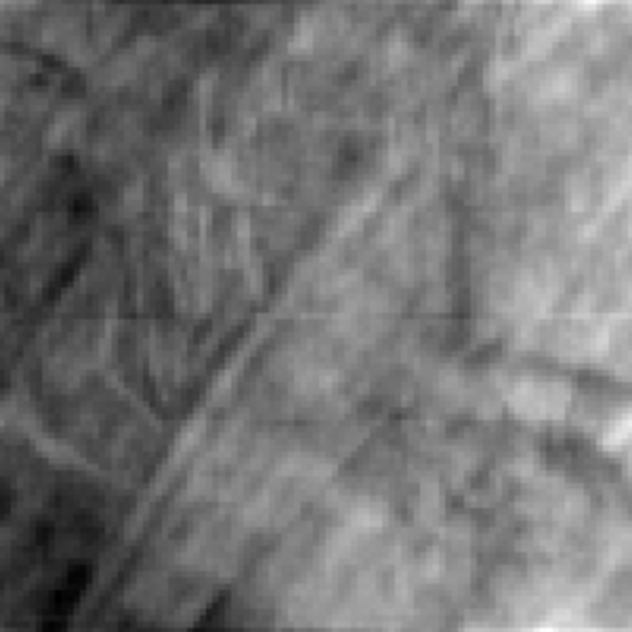}
 }
 
 \subfloat[60 Circulants]{
  \includegraphics[width=\samplespacewidth\textwidth]{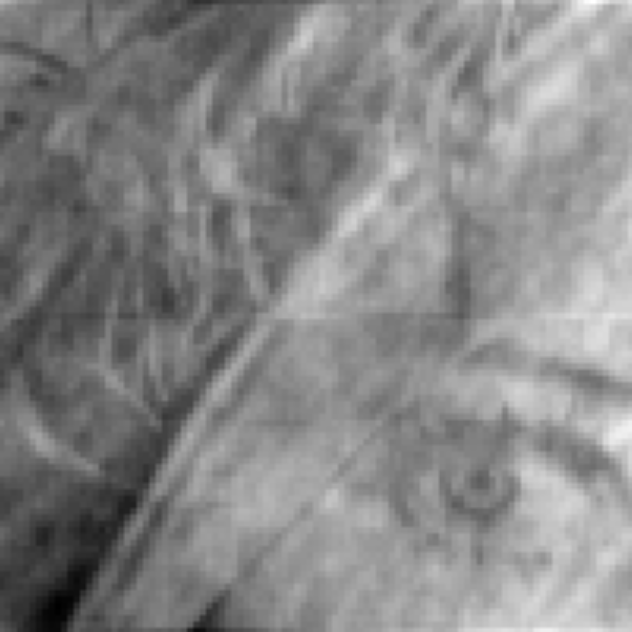}
 }
 \hspace{\samplehspace}
 \subfloat[80 Circulants]{
  \includegraphics[width=\samplespacewidth\textwidth]{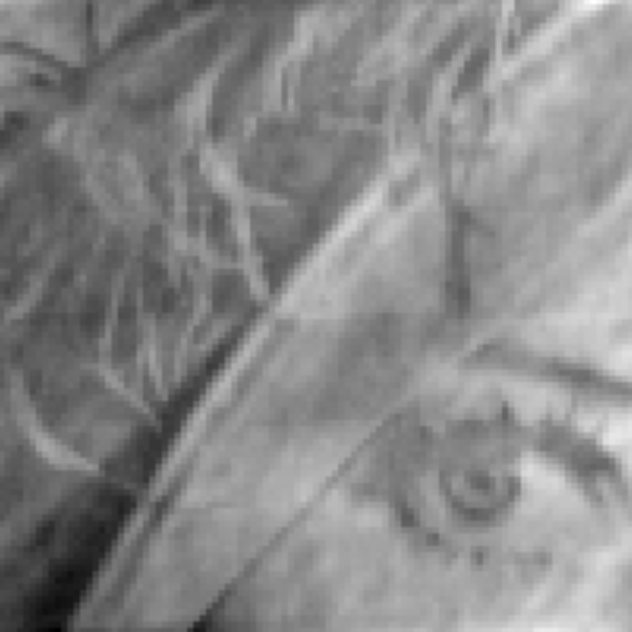}
 }
 \hspace{\samplehspace}
 \subfloat[100 Circulants]{
  \includegraphics[width=\samplespacewidth\textwidth]{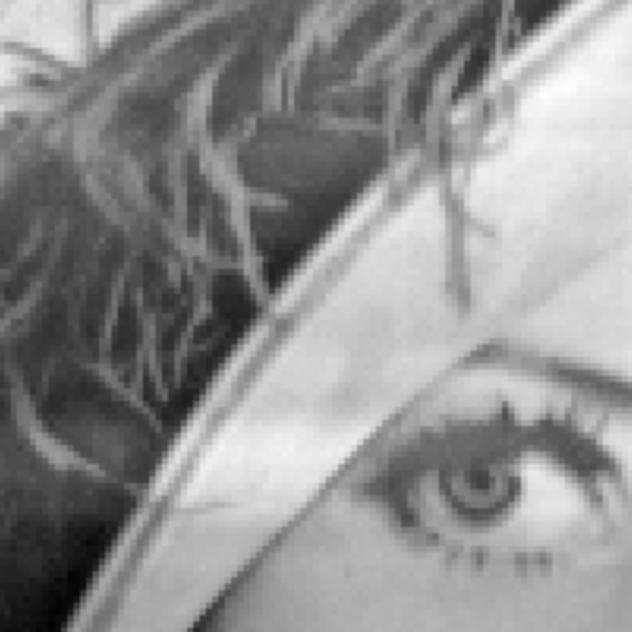}
 }
 \caption{Illustration of the \acl*{CBP} of a $101\times 101$ image of Lena using different number of circulants. (a)-(f) show the effect of an increasing number of slices placed into \acl*{DFT} space on the image. The artefacts on the partial reconstructions are Finite Ghosts, which are the topic of \secTag~\ref{sec::FiniteGhosts}. Using all (102) projections results in a perfect reconstruction of Lena.}
 \label{fig::CBP}
\end{figure}
This is made possible by the prime-sized space of the \ac{DFT}/image, since the $\gcd(m,N) = 1$ always.

The \ac{DRT} was independently developed by Grigoryan~\citep{Grigoryan1986}, Bolker~\citep{Bolker1987}, Gertner~\citep{Gertner1988} and Fill~\citep{Fill1989}. The partitioning allows for the exact $O(n\log_2 n)$ (with $n=N^2$) recovery of images from their projections and is known as the Discrete \ac{FST}~\citep{Grigoryan1986, MatusFlusser, Grigoryan2003}. Do and Vetterli~\citep[TIP]{Do2003} utilised the Wavelet Transform on the \ac{DRT} projections for image compression and de-noising. Chandra~\citep{Chandra2010c} extended the \ac{DRT} to the \ac{NTT} using the circulant theory of the \ac{DRT} and showed that the discrete \ac{FST} still applies. This \ac{NRT} will prove critical in the De-Ghosting process later on in this work. 

The \ac{NRT} methods will rely on initial work done by Normand \emph{et al.}~\citep{NormandEtAl}. They generalised the Katz Criterion to arbitrary convex regions by arguing that a region is reconstructable if and only if the size of the Ghost becomes too large to fit within the image region. To show this, they studied the spatial extents of discrete Ghosts, which were determined using Mathematical Morphology. 

\subsection{Ghost Morphology}
Ghost structures are constructed using dilations of \acp{2PSE}. A discrete \ac{2PSE} is the rational vector $[q,p]$ as shown in part (a) and (b) of \figTag~\ref{fig::Dilation}. 
\begin{figure}[\placement]
 \centering
 \subfloat[]{
 \includegraphics[width=0.13\textwidth]{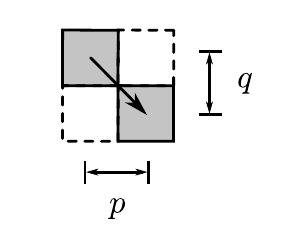}}
 \hspace{0.05cm}
 \subfloat[]{
 \includegraphics[width=0.13\textwidth]{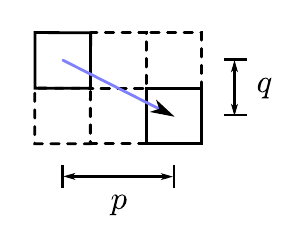}}
 \hspace{0.05cm}
 \subfloat[]{
 \includegraphics[width=0.13\textwidth]{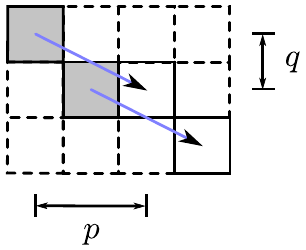}}
 \caption{An example of Minkowski Addition or Morphological Dilation $\MinkowskiSymbol$. (a) and (b) shows \acp{2PSE} based on the vectors $[q, p]$. (c) shows the construction of a parallelogram from the Minkowski addition of (a) with (b).}
 \label{fig::Dilation}
\end{figure}
A dilation by a \ac{2PSE} is the addition of the \ac{2PSE} to every point in the current structure~\citep{Soille2003}. Equivalently, the dilation of $A$ by $B$ is given by the Minkowski Addition of $A$ and $B$ as
\begin{equation}
 A \MinkowskiSymbol B = \{a+b:a\in A,\ b\in B\}.\label{eqn::MinkowskiAddition}
\end{equation}
An example of a dilation is shown in \figTag~\ref{fig::Dilation}. For constructing Ghost structures, each missing projection has a \ac{2PSE} given by their rational angle or vector $[q, p]$, so that $B$ in \eqnTag~\eqref{eqn::MinkowskiAddition} is a \ac{2PSE}. The dilations are then applied using all \acp{2PSE} of the missing projection angles to construct a Ghost structure. 

In the next section, the theory of Finite Ghosts is constructed that describes the artefacts caused by missing slices within the \ac{DFT}. The artefacts formed within the redundant region of a reconstruction can then be used to recover the missing slices.

\section{Finite Ghosts}\label{sec::FiniteGhosts}
The approach of Normand \emph{et al.}~\citep{NormandEtAl} is only adequate when one is concerned with the spatial extent of the Ghosts, such as in defining the reconstruction criterion. When one is also concerned about the values of the Ghosts, the dilation operation must be coupled with a negation of the signs within the Ghost structures. The negation ensures that the Ghost is also zero-valued in the direction of the \ac{2PSE} (see \figTag~\ref{fig::Kernel}). 
\begin{figure}[\placement]
 \centering
 \subfloat[]{
 \includegraphics[width=0.13\textwidth]{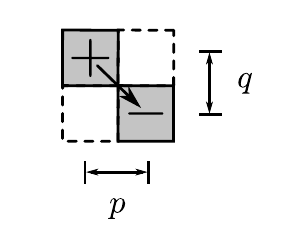}}
 \hspace{0.05cm}
 \subfloat[]{
 \includegraphics[width=0.13\textwidth]{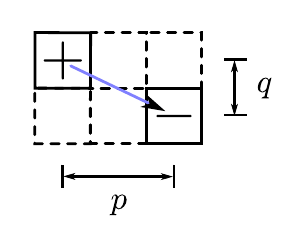}}
 \hspace{0.05cm}
 \subfloat[]{
 \includegraphics[width=0.13\textwidth]{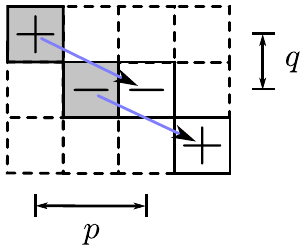}}
 \caption{An example of Ghost convolution kernels as an extension of \figTag~\ref{fig::Dilation}. (a) and (b) shows kernels based on the vectors $[1,m_\firstIndex]$ for each of the missing projections. (c) shows the construction of a parallelogram from the \acs*{2D} convolution of (a) with (b).}
 \label{fig::Kernel}
\end{figure}
This allows Ghost construction to become a convenient \ac{2D} convolution. In the case of the \ac{DFT}, the resulting operation is a \ac{2D} cyclic convolution and the discrete Ghosts are cyclic (or finite) also~\citep{Chandra2008b}.
\begin{theorem}[Finite Ghosts]\label{thm::Ghosts}
Each missing projection $\mat{a}$ at $[1,m_a]$ in the \ac{DRT}, which corresponds to a missing slice in the \ac{DFT}, forms artefacts superimposed on the reconstructed image in the form of a circulant. The unique row of the circulant is $-\mat{a}$. These artefacts are called Finite Ghosts.
\end{theorem}
\begin{proof}
When projections are missing, the \ac{CBP} is incomplete since there must be $\numberSymbol + \nicefrac{\numberSymbol}{\primeSymbol}$ projections for an exact reconstruction. Each projection is equivalent to a circulant in the reconstruction process. Thus, the remaining missing contributions must be a superposition of $\ghostNumberSymbol$ number of circulants with shifts $m_\firstIndex$, where $\firstIndex = 0,\ldots,\ghostNumberSymbol-1$ and $\ghostNumberSymbol$ represents the number of missing projections. The Ghosts are negative-valued since they are missing contributions in the reconstruction.
\end{proof}
These Ghosts can be seen in \figTag~\ref{fig::CBP} as slices are placed into \ac{DFT} space. A schematic of the circulant form of the these Ghosts is shown in \figTag~\ref{fig::Ghosts_Zoomed}. 
\begin{figure}[\placement]
 \centering
 \includegraphics[width=0.4\textwidth]{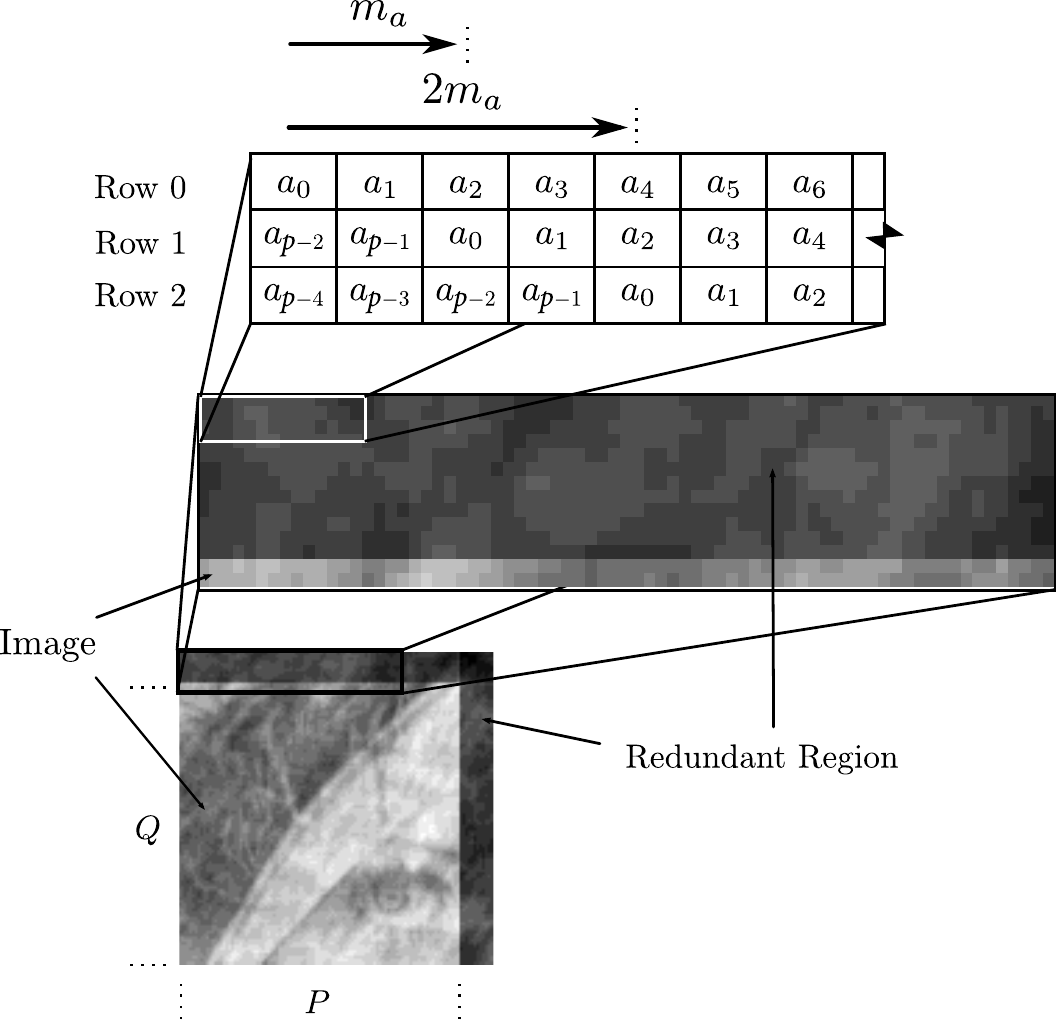}
 \caption{Ghosts in image space when missing projections are present in the \acl*{DRT}. The $Q\times P$ image (Lena) is of $100\times 100$ pixels embedded in a image space $\numberSymbol=113$. Assuming one missing projection $\mat{a}$ at $m_a=2$, then the table shows the circulant nature of the Ghost artefacts that become embedded over the reconstructed image values.}
 \label{fig::Ghosts_Zoomed}
\end{figure}
The circulants effectively represent a \ac{2D} convolution of the missing projection based on the vectors $[1,m]$ of the \ac{DRT}. In the next few sections, a method for defining this convolution is made and then a method for de-convolving or De-Ghosting these Ghosts is constructed.

\subsection{Ghost Convolution}
Let the \ac{2PSE} of the \ac{DRT} projections be defined as $[1,m_\firstIndex]$ for each of the $\firstIndex$ projections. 
\begin{theorem}[Ghost~\acp{2PSE}]\label{thm::Ghost2PSEs}
 The \acp{2PSE} of the Finite Ghosts are $[1,m_\firstIndex]$ for each of the missing projections $\firstIndex = 0,\ldots,\ghostNumberSymbol-1$.
\end{theorem}
\begin{proof}
 From~\propTag~\ref{thm::Ghosts}, each missing projection is a circulant shifted by $m_\firstIndex$ and hence each artefact is shifted by $m_\firstIndex$ on subsequent rows. Therefore, this is equivalent to the Minkowski Addition of each row by $[1,m_\firstIndex]$.
\end{proof}
\begin{theorem}[Ghost Kernels]\label{thm::GhostKernels}
 The Ghost convolution kernels are the \acp{2PSE} $[1,m_\firstIndex]$, which sum to zero along its vector for each of the missing projections $\firstIndex = 0,\ldots,\ghostNumberSymbol-1$.
\end{theorem}
\begin{proof}
The projections of Ghosts are zero-valued at certain finite angles $[1,m_\firstIndex]$ by construction (see \propTag~\ref{thm::Ghost2PSEs}), so the \acp{2PSE} must also be a Ghost at its finite angle, since they apply to each point in the image/structure. These then become Ghost convolution kernels (see \eqnTag~\eqref{eqn::MinkowskiAddition}). An example of Ghost kernels can be seen in \figTag~\ref{fig::Kernel}. 
\end{proof}

Then the Ghost kernels can be used in an initial method to construct Finite Ghosts by computing the \ac{2D} convolution of the Ghost kernels as follows:
\begin{enumerate}
 \item Compute the \ac{DFT} of the initial structure and each of the \acp{2PSE}.
 \item Multiply the coefficients of each \ac{2PSE}, together with the coefficients of the structure, in \ac{DFT} space to compute the \ac{2D} convolution.\label{enum::Convolve}
 \item Compute the inverse \ac{DFT} to get the Ghosts in image space.
\end{enumerate}
The method has a computational complexity of $O(\ghostNumberSymbol\numberSymbol^2\log_2 \numberSymbol)$ for computing the \ac{2D} \acp{DFT} and $O(\ghostNumberSymbol\numberSymbol^2)$ for computing the convolution, where $\ghostNumberSymbol$ denotes the number of Ghosts. An example of the \ac{2D} convolution method is given in \figTag~\ref{fig::Kernel}. The resulting Ghost will be invisible for those projections at the finite angles used to construct them.

However, this \ac{2D} convolution approach is inefficient because it does not take into account the following:
\begin{enumerate}
 \item The Ghost grows a row at a time and, in most cases, does not take up the full \ac{2D} space. Thus, a series of full \ac{2D} convolutions can be expensive.
 \item The \ac{DRT} space contains zero valued projections and when $\ghostNumberSymbol$ is close to $\numberSymbol$, this space is mostly empty.
\end{enumerate}
Is there a way to make the process more efficient? Yes, the answer is to use the \acl{PCT} of the \ac{DRT}.

\subsection{Projection Convolution Theorem}
The \ac{PCT} states that a \ac{2D} convolution is equivalent to the \ac{1D} convolution of each projection or slice in Fourier space. Thus, the \ac{2D} Ghost convolution can be computed as a series of \ac{1D} convolutions on the known projections. For Ghost convolution, this is particularly useful when the number of Ghosts $\ghostNumberSymbol$ is close to $\numberSymbol$, so that the number of known projections is small. Therefore, \ac{DRT} space is sparse and the \ac{2D} convolution becomes a relatively small number of \ac{1D} operations.

The \ac{PCT} can be interpreted as a consequence of the discrete \ac{FST}. Since the slices tile \ac{DFT} space exactly and that slices are the \ac{DFT} of the projections, cyclically convolving the slices of two objects is equivalent to cyclically convolving the objects themselves. Thus, in order to utilise the \ac{PCT}, one needs to know the projections of the two objects. In this case, the projections of the Ghost kernels are particularly convenient, which makes a \ac{PCT} approach very efficient.
\begin{theorem}[Kernel Projections]\label{thm::KernelProjections}
The projections of the \ac{2D} convolution kernel $[1,m_\ghostSymbol]$, with a positive term at the origin and a negative term at the coordinate $(1,m_\ghostSymbol)$, will have the positive term at the zeroth translate and the negative term at translate $t_\firstIndex$ as
\begin{equation}
t_\firstIndex = \left(m_\ghostSymbol - m_\firstIndex\right) \imod\numberSymbol\label{eqn::Kernels}
\end{equation}
for each projection $m_\firstIndex$ with $\firstIndex = 0,\ldots,\numberSymbol-1$ and $\ghostSymbol = 0,\ldots,\numberSymbol-1$. For the $\firstIndex = \numberSymbol$ projection, negative term is always at $t_\firstIndex = 1$. Note that negative values $-x$ modulo $\numberSymbol$ are equivalent to the value $\numberSymbol-x \imod \numberSymbol$.
\end{theorem}
\begin{proof}
The projection $m_\ghostSymbol$ will be zero by definition with all other projections being non-zero. For projections with $m_\firstIndex < m_\ghostSymbol$, the negative term will appear in the translates $t>0 \!\imod \numberSymbol$ because the initial \ac{2PSE} $[1,m_\firstIndex]$ for $t=0$ will not sample the point $(1,m_\ghostSymbol)$ while projecting. The difference of slope $(m_\ghostSymbol - m_\firstIndex)$ defines the translate where the point is eventually sampled. For projections with $m_\firstIndex > m_\ghostSymbol$, the negative term will appear in the translates $t<0 \!\imod \numberSymbol$ because the initial \ac{2PSE} $[1,m]$ will not sample the point $(1,m_\ghostSymbol)$ while projecting until the \ac{2PSE} wraps around the image. Hence, the difference of slope $(m_\ghostSymbol - m_\firstIndex)$ will be negative and defines the translate where the point is eventually sampled.
\end{proof}
\begin{corollary}[Kernel Operators]\label{thm::KernelOps}
The projections of the kernels $[1,m_\ghostSymbol]$ only have $\numberSymbol$ distinct combinations of positive and negative values.
\end{corollary}
\begin{proof}
This follows from \propTag~\ref{thm::KernelProjections} and the fact that the system is constructed within a finite geometry, i.e. because of the wrapping of the values since there are only $\numberSymbol$ residue classes.
\end{proof}
Thus, the unique kernel projections (following \corTag~\ref{thm::KernelOps}) and their eigenvalues have the form as shown in part (a) of \figTag~\ref{fig::GhostOpsSchematic}.
\begin{figure}[\placement]
 \centering
 \subfloat[]{
 \includegraphics[width=0.21\textwidth]{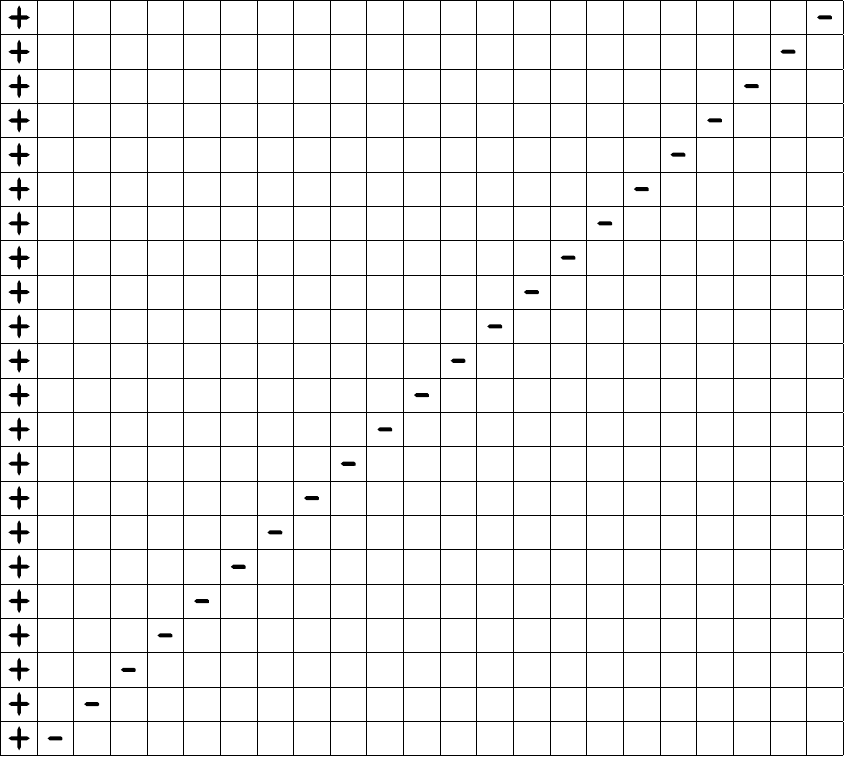}}
 \hspace{0.15cm}
 \subfloat[]{
 \includegraphics[width=0.2\textwidth]{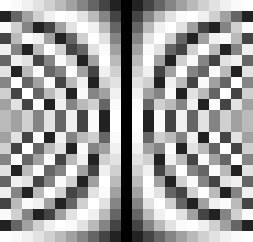}}
 \caption{A schematic of the \ac{1D} Ghost convolution operators (a). Each row represents the unique projection of a Ghost convolution kernel. (b) shows their eigenvalues in (centred) discrete Fourier space.}
 \label{fig::GhostOpsSchematic}
\end{figure}
These eigenvalues can be precomputed, since they depend purely on the set of finite angles of the missing projections as given by \propTag~\ref{thm::KernelProjections}. The form of the \ac{DFT} eigenvalues are shown in part (b) of \figTag~\ref{fig::GhostOpsSchematic} and \figTag~\ref{fig::GhostOperators} for $\numberSymbol=479$.

Thus, the \ac{1D} Ghost convolution approach may be used to generate a \ac{2D} Ghost as follows:
\begin{enumerate}
 \item Pre-compute the Ghost operator eigenvalues via the \ac{DFT}. This can be used as a hash table to pick out the relevant operator based on the projection being convolved.
 \item Convolve a delta function with each of the \acp{2PSE} by selecting the correct eigenvalues for each operator using the eigenvalues hash table (as given by previous step) via \eqnTag~\eqref{eqn::Kernels} and multiplying these operators with the eigenvalues of the delta function.
 \item Inverse \ac{DFT} to obtain the Ghost structure in image space.
\end{enumerate}
The computational complexity is $O(\ghostNumberSymbol \numberSymbol\log_2 \numberSymbol)$ for pre-computing the eigenvalues, $O(\mu \numberSymbol\log_2 \numberSymbol)$ for the \ac{1D} \acp{DFT} and $O(\mu \ghostNumberSymbol \numberSymbol)$ for computing the convolutions. This method is advantageous when $\mu \ll \numberSymbol$, since the computations are done only on the small number of known projections. A simple example of this construction is shown in \figTag~\ref{fig::GhostOperatorsEg} for two \acp{2PSE}.
\begin{figure}[\placement]
 \centering
 \subfloat[]{
 \includegraphics[width=0.10\textwidth]{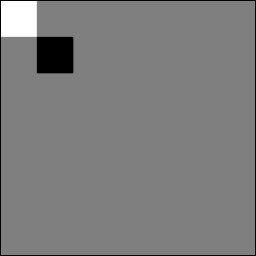}}
 \hspace{0.1cm}
 \subfloat[]{
 \includegraphics[width=0.11\textwidth]{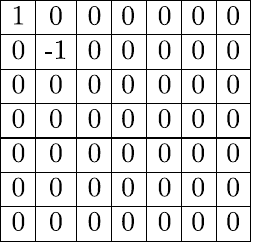}}
 \hspace{0.1cm}
 \subfloat[]{
 \includegraphics[width=0.10\textwidth]{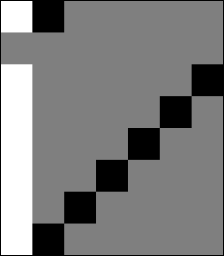}}
 
 \subfloat[]{
 \includegraphics[width=0.10\textwidth]{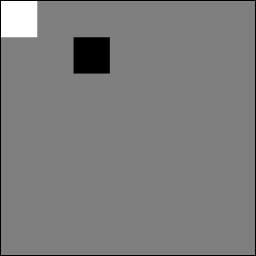}}
 \hspace{0.1cm}
 \subfloat[]{
 \includegraphics[width=0.11\textwidth]{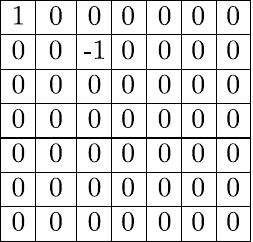}}
 \hspace{0.1cm}
 \subfloat[]{
 \includegraphics[width=0.10\textwidth]{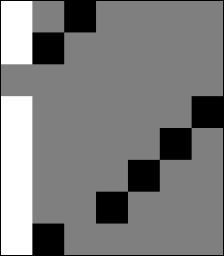}}
 
 \subfloat[]{
 \includegraphics[width=0.10\textwidth]{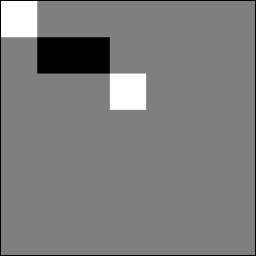}}
 \hspace{0.1cm}
 \subfloat[]{
 \includegraphics[width=0.11\textwidth]{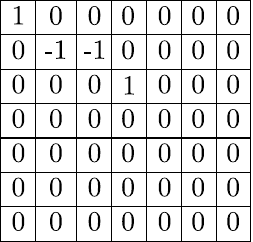}}
 \hspace{0.1cm}
 \subfloat[]{
 \includegraphics[width=0.10\textwidth]{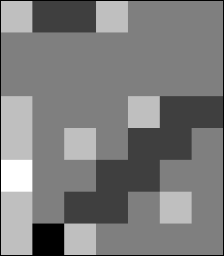}}
 \caption{A simple example of the \ac{1D} convolution approach to constructing \ac{2D} Ghosts, where grey and black denote zero and minus greyscale values respectively. (a) and (d) show the \acp{2PSE} $[1,1]$ and $[1,2]$ respectively. (c), (f) and (i) show the \ac{DRT} projections of the \ac{2PSE} (a), the \ac{2PSE} (d) and the convolution of \acp{2PSE} (g) respectively.}
 \label{fig::GhostOperatorsEg}
\end{figure}

However, both methods are susceptible to round-off errors and loss of precision for large $\ghostNumberSymbol$ when utilising the \ac{DFT} for the convolutions. This problem, which manifested as numerical overflow, was also encountered by Chandra~\emph{et al.}~\citep{Chandra2008} with their method to remove Finite Ghosts. This numerical growth can be easily seen as a direct consequence of the convolutions of the eigenvalues in \figTag~\ref{fig::GhostOperators}. The solution is to use the \acl{NTT} and the \acl{NRT} of Chandra~\citep{Chandra2010c}.

\subsection{Number Theoretic Convolution}
The \acf{NTT} allows one to compute convolutions just as the \ac{DFT} because the unit circle is replaced with the digital ``circle''
\begin{eqnarray}\label{eqn::Euler}
\primRootSymbol^{\modulusSymbol-1} &\equiv& 1 \imod \modulusSymbol,
\end{eqnarray}
so that $\primRootSymbol^{\modulusSymbol-1} - 1$ is a multiple of $\modulusSymbol$, $\numberSymbol$ is a multiple of $\modulusSymbol-1$ and $\primRootSymbol, \modulusSymbol \in \wholeNumbers$~\citep{Pollard1971}. The successive powers $\{1,\ldots,\modulusSymbol-1\}$ of $\primRootSymbol$ generates a unique set of integers in some order modulo $\modulusSymbol$. Such a number $\primRootSymbol$ is called the primitive root of $\modulusSymbol$. 

The primitive root(s) $\primRootSymbol$ in these cases have to be found by trial and error and can be computed by dividing $\modulusSymbol-1$ by the prime factors $\primeSymbol_j$ of $M-1$, such that $\primRootSymbol^{(\modulusSymbol-1)/\primeSymbol_j} \not\equiv 1 \imod \modulusSymbol$, where the trial value of $\primRootSymbol$ is prime. Integer coefficients allow computations to be done without round-off error or numerical overflow, since the results are congruent modulo $\modulusSymbol$~\citep{Nussbaumer1978}. A new efficient and fast algorithm for computing prime-length \acp{NTT}, that may be utilised for this work, is described in~\ref{sec::FastNTT}.

Chandra~\citep{Chandra2010c} showed that the discrete \ac{FST} still holds within the \ac{2D} \ac{NTT} when placing Number Theoretic slices, i.e. the \acp{NTT} of the projections, into \ac{2D} \ac{NTT} space. This allows one to replace the \ac{DFT} in all computations, including those within the Ghost convolutions, with the \ac{NTT}. The resulting \ac{NRT} was constructed specifically to remedy the loss of precision when forming Finite Ghosts. Chandra~\citep{Chandra2010c} also showed that the implementation of the \ac{NTT} is faster than the \ac{DFT} because of its integer-only operations. 

Consequently, the Ghost convolution method is impervious to numerical overflow and faster than \ac{DFT} approaches. The \ac{NTT} integer eigenvalues of the Ghosts kernel projections are shown in \figTag~\ref{fig::GhostNTTOperators}.
\begin{figure}[\placement]
 \centering
 \includegraphics[width=0.45\textwidth]{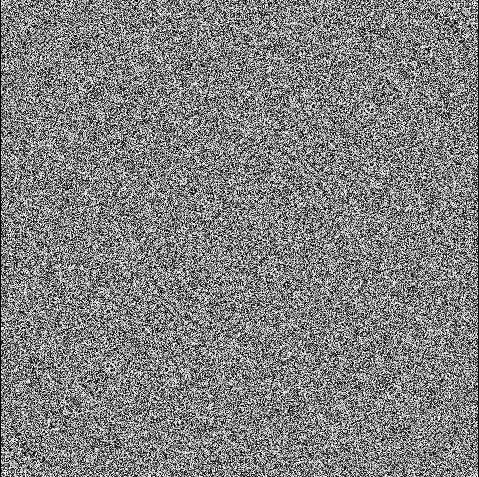}
 \caption{The \acl*{NTT} eigenvalues of one dimensional circulant Ghost operators for $\numberSymbol=479$. Note that the fine structure present is more easily seen when the above image is viewed from a distance.}
 \label{fig::GhostNTTOperators}
\end{figure}
The Ghosts in this space have the value $0 \!\!\imod \modulusSymbol$ in the direction of the missing projections. An example of Ghosts within the \ac{NRT} can be seen in \figTag~\ref{fig::NRTGhosts}.
\begin{figure}
\centering
\subfloat[]{
\includegraphics[width=0.15\textwidth]{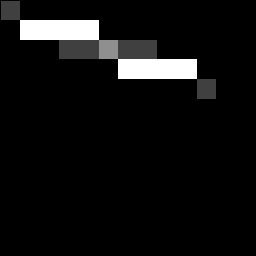}}
\hspace{0.75cm}
\subfloat[]{
\includegraphics[width=0.2\textwidth]{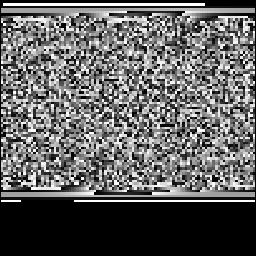}}
\caption{Examples of Ghosts in the \acl*{NRT}. (a) shows $\ghostNumberSymbol=4$ within a $\numberSymbol=13$ with Ghosts at $m=1,2,3,4 \imod {53}$. (b) shows $\ghostNumberSymbol=80$ within a $\numberSymbol=101$ with Ghosts at $m=1,\ldots,80 \imod {607}$. Each Ghost occupies $\ghostNumberSymbol+1$ rows in the image.}
\label{fig::NRTGhosts}
\end{figure}
Ghosts in the \ac{NRT} have another important property in that any physical or perceivable structures in the Ghosts are difficult to discern. This makes the \ac{NRT} Ghosts well suited for encoding and encryption. 

In summary, the advantages of the \ac{1D} Ghost convolution approach are
\begin{enumerate}
 \item Simplicity: \ac{1D} operators are simple and cyclic convolutions are much easier to compute.
 \item Efficiency: The majority of \ac{DRT} or \ac{NRT} space will be empty due to the Ghosts, so one always keeps the amount of work to a minimum.
 \item Speed: \ac{1D} blocks are at most $\numberSymbol$ elements apart in memory. 
\end{enumerate}

Once the Ghost is constructed/convolved, it can be used to remove the artefacts formed by the missing slices (represented via their \acp{2PSE}) from which it is constructed. This process of de-convolution is described in the next section.

\subsection{Ghost De-Convolution}\label{sec::Deconvolve}
The Ghost convolved in the previous section may be applied as a mask to remove the Ghost introduced by missing projections in the \ac{DRT} or missing slices in the \ac{DFT}. The mask is applied by individually convolving each of the rows of the Ghost with the rows of the image as shown in \figTag~\ref{fig::SubCirc}. 
\begin{figure}[\placement]
 \centering
 \includegraphics[width=0.33\textwidth]{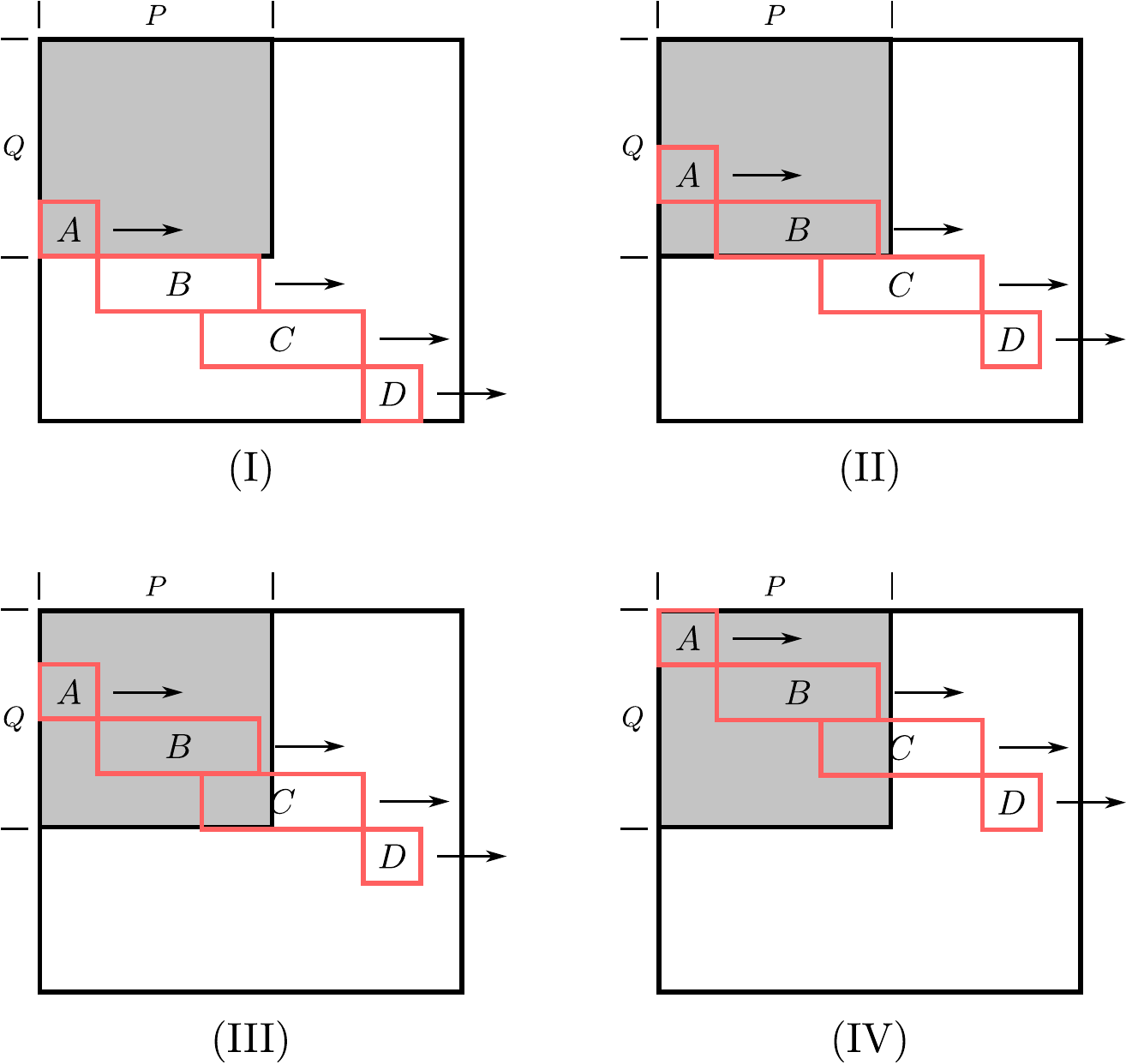}
 \caption{To partially de-convolve the $Q\times P$ sub-image (in grey), one aligns the first row $A$ of the Ghost constructed (in red) to the image row to be De-Ghosted. The remaining rows would be in the redundant region. Once (I) is completed, the De-Ghosted sub-image row is removed from the Ghosts in (II) to make it into a redundant row. The process is continued until all image rows are De-Ghosted.}
 \label{fig::SubCirc}
\end{figure}

Part (I) of \figTag~\ref{fig::SubCirc} is effectively a sum of \ac{1D} convolutions or a sub-circulant expression given as 
\begin{equation}
\mat{A}\cdot \tilde{I}(3,0)^T + \mat{B}\cdot \tilde{I}(4,0)^T + \mat{C}\cdot \tilde{I}(5,0)^T + \mat{D}\cdot \tilde{I}(6,0)^T = I(3,0)^T,\label{eqn::SubCirc}
\end{equation}
where $\mat{A}, \mat{B}, \mat{C}$ and $\mat{D}$ are circulants, $\tilde{I}$ denotes the Ghosted image and $I$ denotes the desired De-Ghosted sub-image. Row $I(3,0)$ is back-substituted  into $\tilde{I}(3,0)$ order to make it a redundant row. Then \eqnTag~\eqref{eqn::SubCirc} is repeated for (II) and so on. In this method, it is important to note that the shifts need to be $-m_\ghostSymbol$, rather than $m_\ghostSymbol$, when convolving Ghosts to undo the right shift of the circulants. This is done by setting $N-m_\ghostSymbol$ as the projection angle for the missing projections $m_\ghostSymbol$ in the Ghost convolutions.

The coefficients of the Ghosted image and Ghost structure need only be computed once and the remaining computations done in \ac{DFT} or \ac{NTT} space. Thus, the computational complexity of this de-convolution method is $O(Q\ghostNumberSymbol\primeSymbol)$. Note also that the de-convolution requires a total of $\ghostNumberSymbol+1$ rows, with $\ghostNumberSymbol$ of those rows being redundant, to function.

A visual interpretation to constructing Ghosts, that unifies the Ghost Recovery algorithm of Chandra~\citep{Chandra2008} and the \ac{2D} convolution approach of this paper, can be made using $n$-gons. See the thesis of Chandra~\citep{Chandra2010d} for initial work on this topic.

More work has to be done with both methods when \ac{DFT} space has noise or inconsistencies present as these also become convolved during the De-Ghost process. The convolution approach requires a very good estimation of noise prior to De-Ghosting, so that the estimates may be used to de-convolve their effects on the results. Further work needs to be done in generalising the convolution approach to arbitrary missing discrete Fourier coefficients using the theory outlined in this work. Recent work by Svalbe \emph{et al.}~\citep{Svalbe2010b} discussed the minimal extent of Finite Ghosts, which may prove useful in this endeavour. In the next section, the De-Ghost method is applied to the discrete inverse problem of determining a reconstruction from a set of rational angle aperiodic projections.

\section{Discrete Reconstruction}\label{sec::Mojette}
Chandra \emph{et al.}~\citep{Chandra2008} utilised their Ghost removal technique to reconstruct an image from a limited angle set of discrete \ac{1D} projections. The term limited angle refers to the fact that the coverage of the half-plane need not be uniform and/or complete. Recent work done by Chandra~\emph{et al.}~\citep{Chandra2010a} has been done in improving the use of the discrete projection data within the \ac{DFT} for fast reconstruction. However, the angle set utilised in~\citep{Chandra2010a} covered the half-plane. In this section, the De-Ghost method will be used to extend the work of Chandra~\emph{et al.}~\citep{Chandra2010a} to include any limited angle set of projections and reducing the total number of unique projections (for an exact reconstruction) to $Q+1$ for a $Q\times P$ image.

For a $Q\times P$ image in a $\numberSymbol\times \numberSymbol$ space, a total of $Q+1$ projections need to be taken because one needs $\ghostNumberSymbol$ redundant rows in the image for $\ghostNumberSymbol$ missing slices. This means that there must be $\mu$ number of projections so that $\ghostNumberSymbol = \numberSymbol-\mu = \numberSymbol-(Q+1)$, since there are $\numberSymbol+1$ total number of slices. The rational angle projection geometry is known as the \ac{MT}. \figTag~\ref{fig::MT} shows a simple example of a \ac{MT} for a $4\times 4$ image using three projections.
\begin{figure}[\placement]
 \centering
 \includegraphics[width=0.48\textwidth]{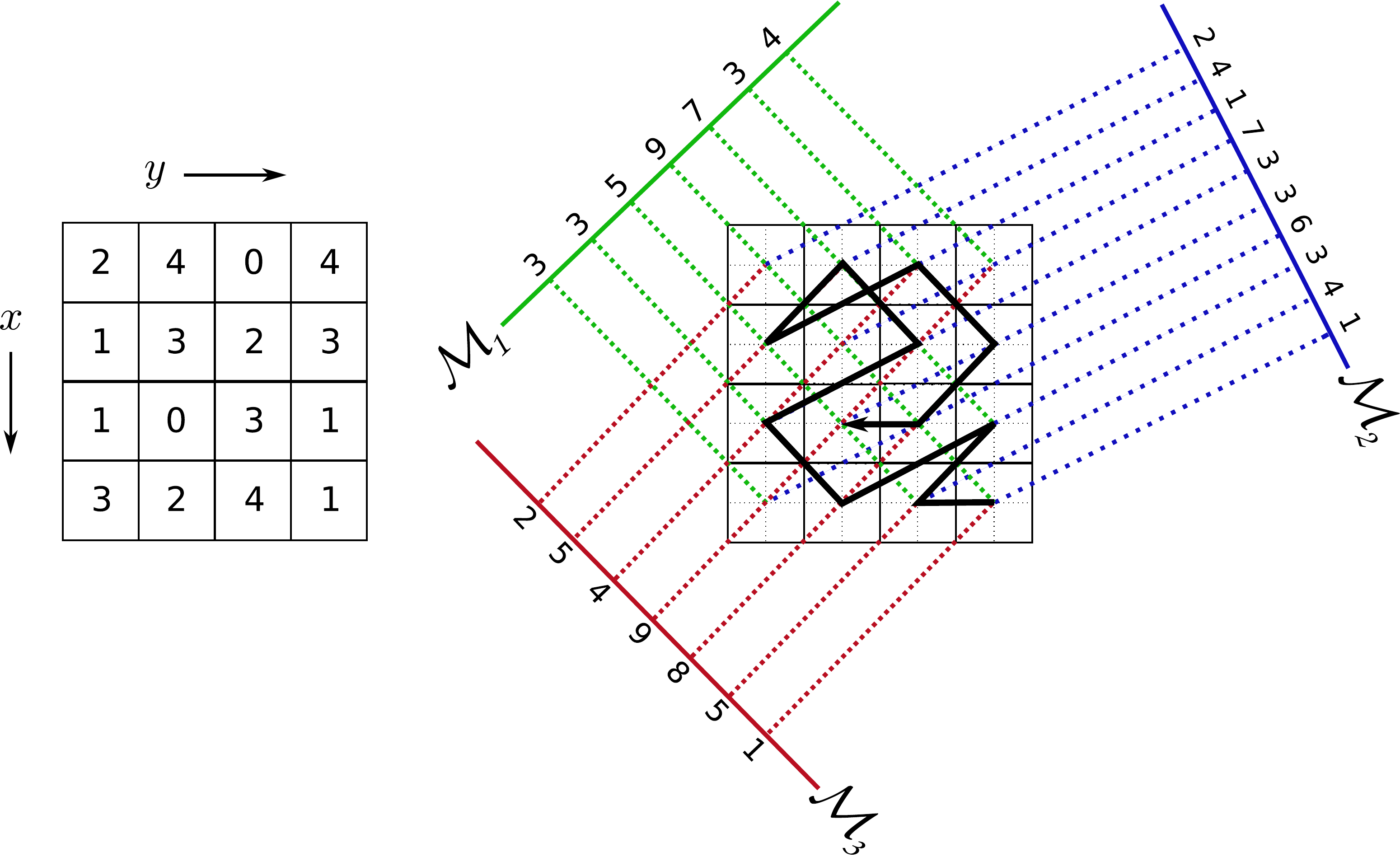}
 \caption{An example of a \acl*{MT} for a discrete image of size $4\times 4$ using the three projections $[1,1]$, $[1,-1]$ and~$[1,-2]$. The bold lines within the right-hand grid shows a possible reconstruction path using a corner-based reconstruction method~\citep{Normand2006}.}
\label{fig::MT}
\end{figure}

\subsection{Ghost Angle Sets}
Chandra~\emph{et al.}~\citep{Chandra2010a} showed that the \ac{MT} projection set directly and efficiently maps to a prime-sized \ac{DFT} space exactly as
\begin{equation}
 m \equiv p q^{-1} \imod N,\label{eqn::mMap}
\end{equation}
where $q^{-1}$ is the multiplicative inverse of $q$ that can be computed easily via the Extended Euclidean algorithm. For the dyadic or power of two case, extra $s$ projections are needed (see Chandra~\citep{Chandra2010c}), which are mapped as
\begin{equation}
 2s \equiv p^{-1} q \imod N,\label{eqn::sMap}
\end{equation}
when the $\gcd(q,N)>1$, Positive $[q,p]$ values represent the first octant of the half-plane with the other octants produced by $[-q,p], [p,q]$ and $[p,-q]$~\citep{Svalbe2011}. The reconstruction is then obtained with a computational complexity of $O(n\log_2 n)$ with $n = N^2$, which is the same order as the \ac{2D} \ac{FFT}. 

The same mapping can be used to cover a fraction of the half-plane by limiting the number of octants used within the projection set. One then obtains a set that is limited, but covers \ac{DFT} space fully because there exists a multiplicity of possible rational vectors for each finite angle $[1,m]$. This multiplicity is shown as graphs in \figTag~\ref{fig::Multiplicity} for the case of a quadrant and a half-plane.
\begin{figure}[\placement]
 \centering
 \subfloat[]{
 \includegraphics[width=0.45\textwidth]{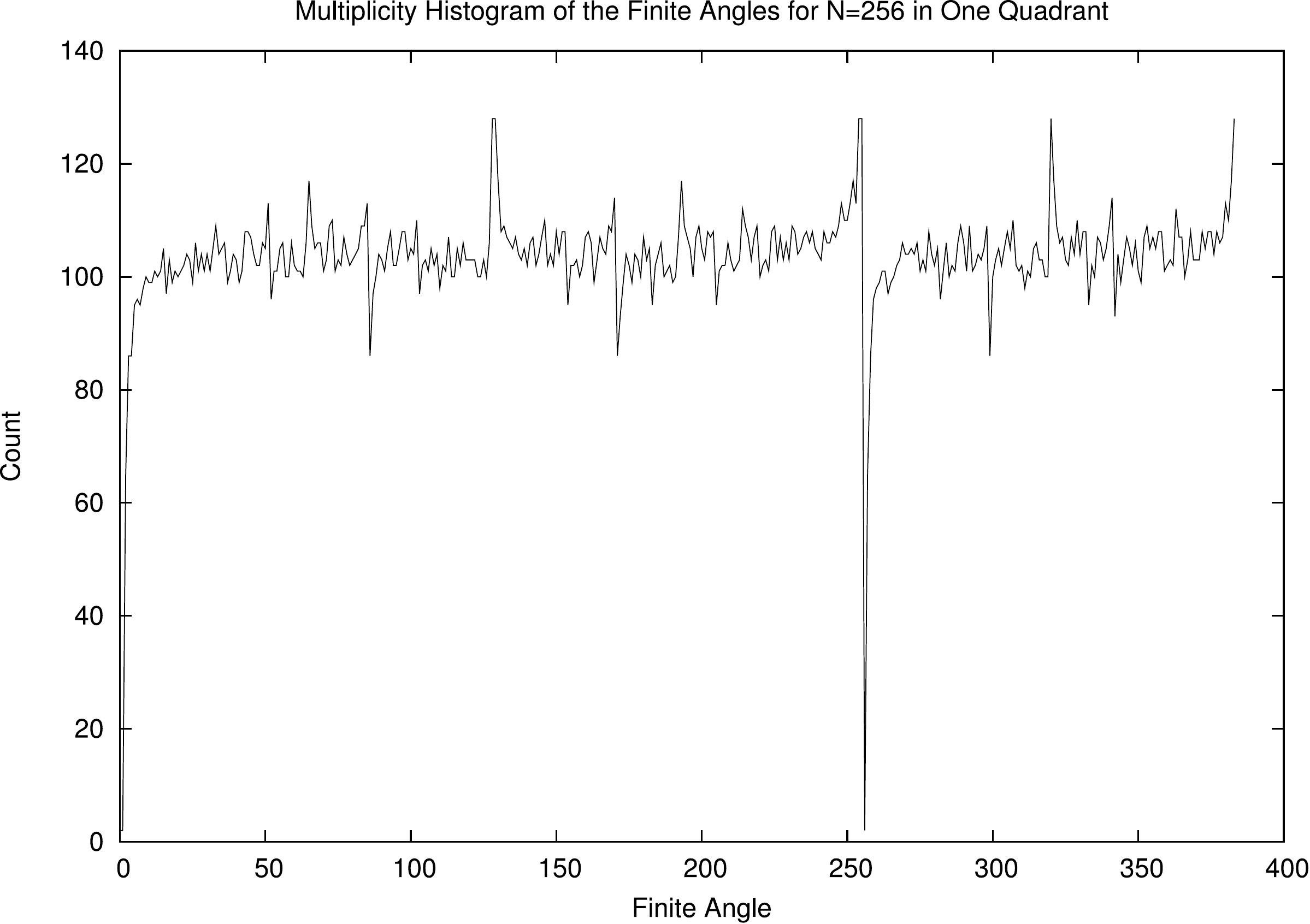}}
 
 \subfloat[]{
 \includegraphics[width=0.45\textwidth]{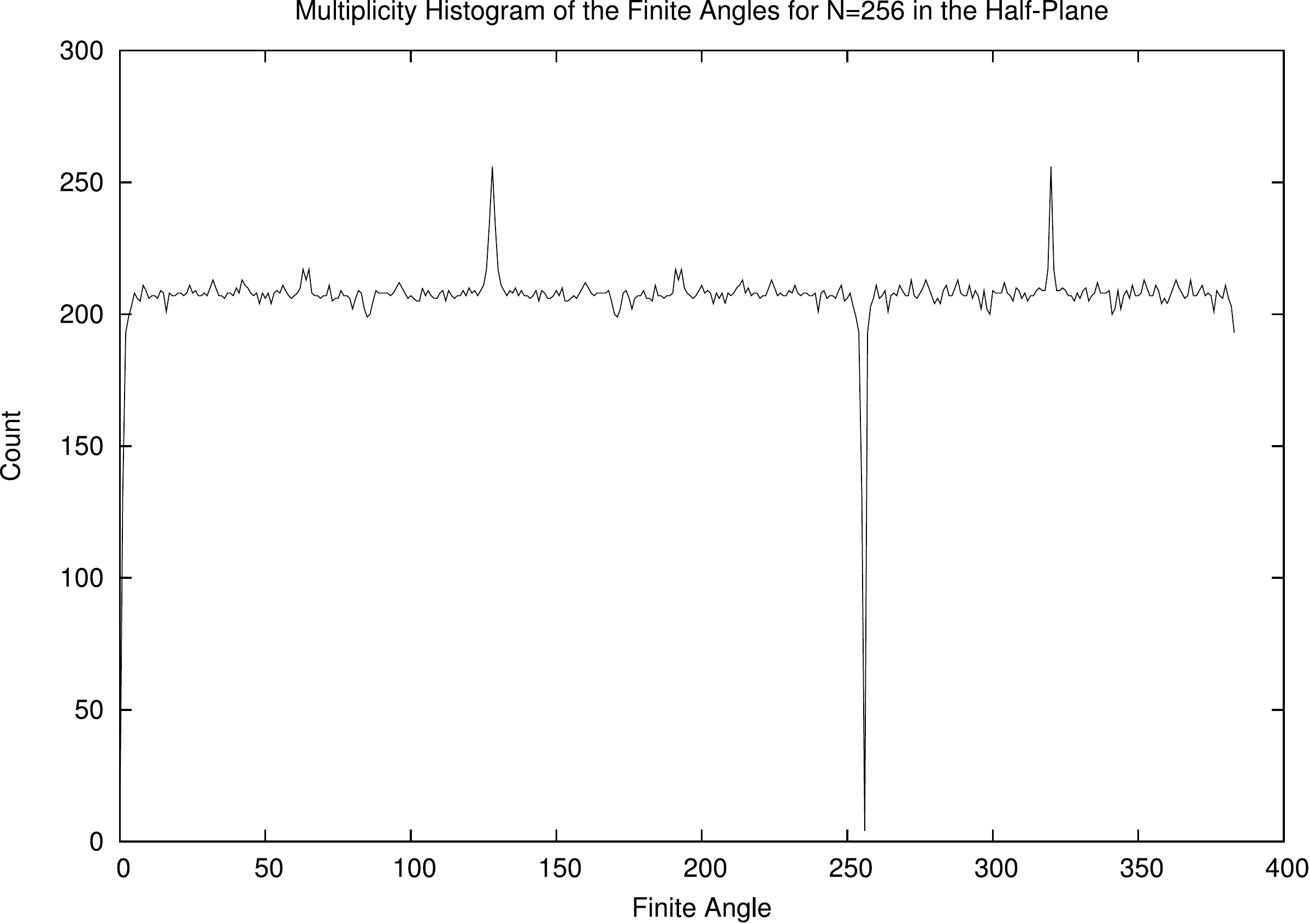}}
 \caption{The multiplicity of the mapping between the finite angle and rational angle sets for $N=256$, i.e. for the dyadic case. Because of the dyadic size, the range of values are $0 \leqslant m < N$ and $s' = N+s$, where $0 \leqslant s < N/2$.}
\label{fig::Multiplicity}
\end{figure}
One may observe that the multiplicity is relatively ``flat'' in a discrete sense, but with small variations. Graph~\ref{fig::Multiplicity}(a) is reminiscent of curves obtained by Svalbe and Kingston~\citep{SvalbeKingston} when observing the ``unevenness'' of the rational vectors from uniform coverage. The unevenness of the Farey vectors is related to the distribution of prime numbers and the Riemann Hypothesis~\citep{Franel1924, Landau1924}.

The theory of Finite Ghosts allows one to reduce the number of projections required to $Q+1$, rather than a total $N+1$, for these limited angle sets by choosing the projections with a desired property, such as the ones with the smallest number of bins. An example of such projection sets are given in \figTag~\ref{fig::GhostSet}.
\begin{figure}[\placement]
 \centering
 \subfloat[]{
 \includegraphics[width=0.18\textwidth]{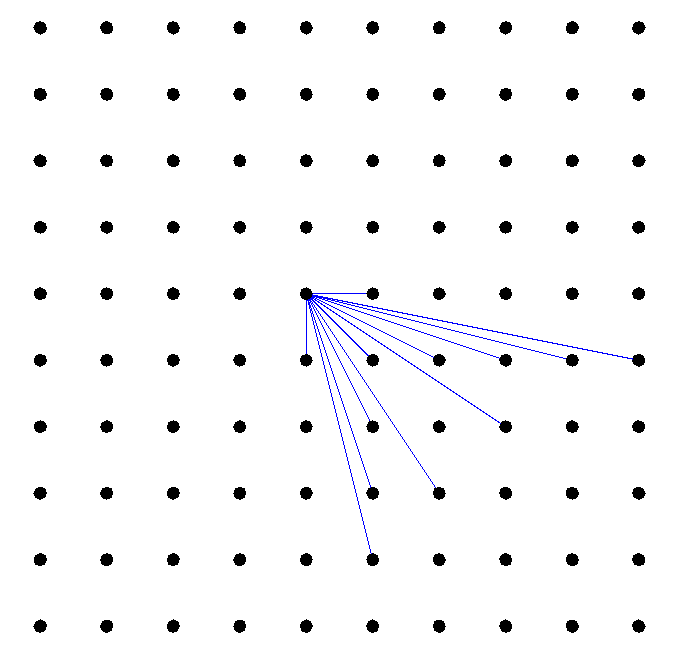}}
 \hspace{0.25cm}
 \subfloat[]{
 \includegraphics[width=0.18\textwidth]{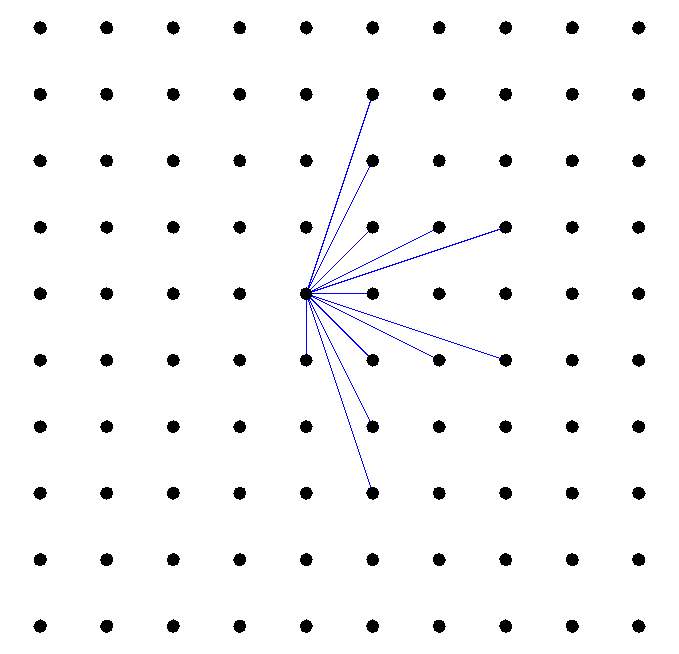}}
 \caption{Examples of rational projection sets when utilising Ghosts in discrete reconstruction. (a) shows a set limited to a quadrant and (b) shows a set using the half-plane. The sets apply to a $11\times 11$ image embedded into a $23\times 23$ space.}
\label{fig::GhostSet}
\end{figure}

\subsection{Results}
Let us consider a $100\times 100$ image of Lena and 101 rational angle projections similar to the geometry of \figTag~\ref{fig::GhostSet}(b). Also, for performance, let $\numberSymbol=257$, so that $\numberSymbol-1$ is a power of two (see \ref{sec::FastNTT}). Therefore, from the 101 projections, there will be 156 Ghosts to remove and also 156 redundant rows in the image, so that the De-Ghost method can be applied.

\figTag~\ref{fig::Results}(a) shows the \ac{DRT} space resulting from 101 rational projections acquired across the half-plane and mapped via \eqnTag~\eqref{eqn::mMap}. \figTag~\ref{fig::Results}(b) and (c) show the resulting Ghosts superimposed on the reconstruction due to the missing projections in (a). \figTag~\ref{fig::Results}(d) shows the \ac{NTT} Ghost eigenvalues that are used to construct the de-convolution Ghost in \figTag~\ref{fig::Results}(e). \figTag~\ref{fig::Results}(f) shows the final result of the De-Ghosting, which is an exact reconstruction.
\begin{figure}[\placement]
 \centering
 \subfloat[]{
 \includegraphics[width=0.16\textwidth]{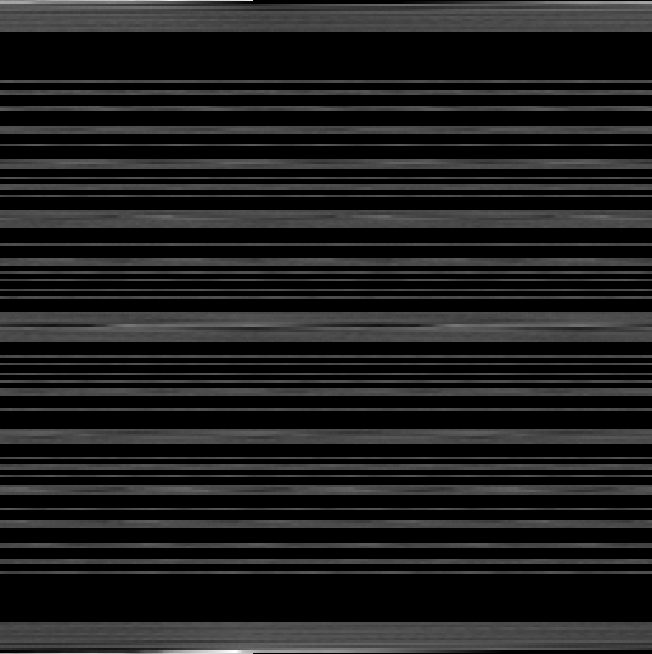}}
 \hspace{0.45cm}
 \subfloat[]{
 \includegraphics[width=0.16\textwidth]{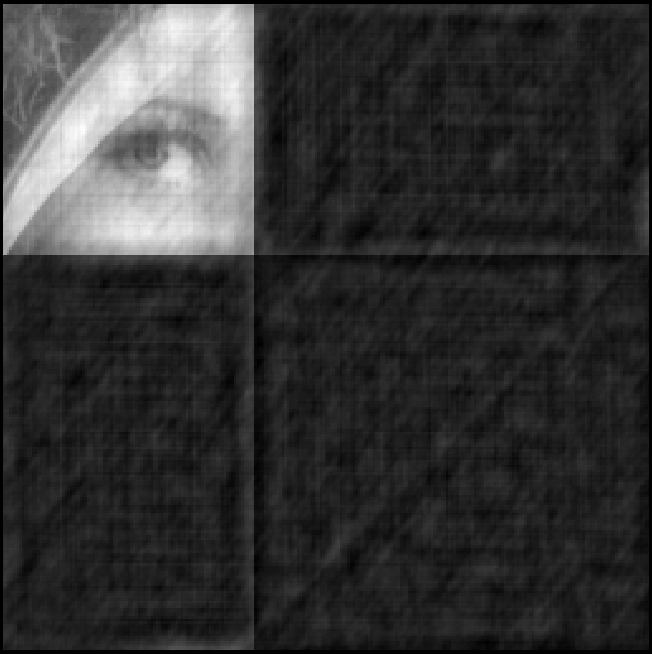}}

 \subfloat[]{
 \includegraphics[width=0.16\textwidth]{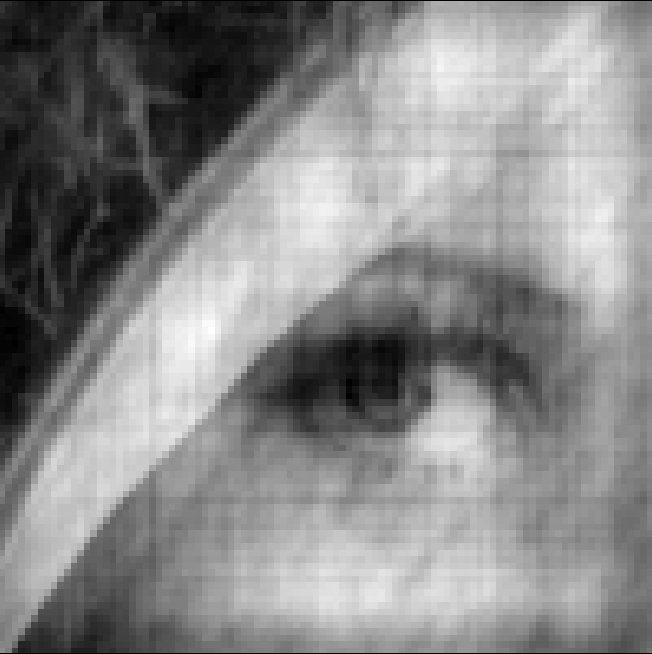}}
 \hspace{0.45cm}
 \subfloat[]{
 \includegraphics[width=0.16\textwidth]{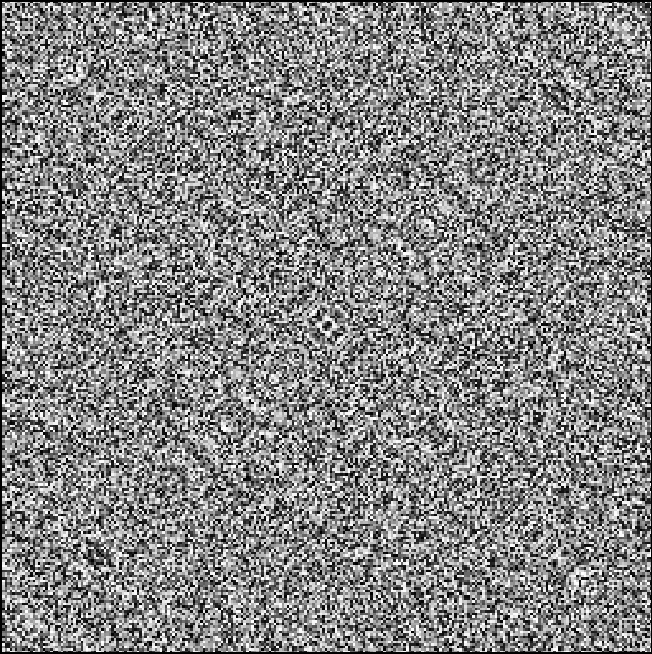}}

 \subfloat[]{
 \includegraphics[width=0.16\textwidth]{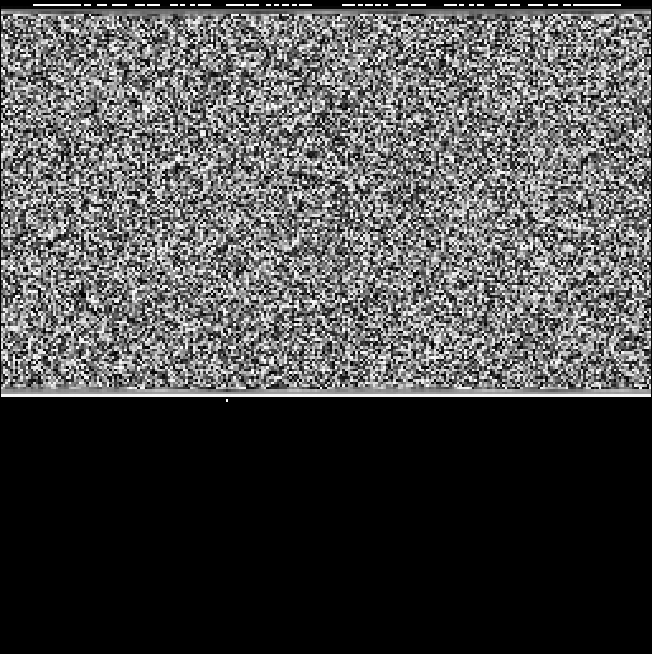}}
 \hspace{0.45cm}
 \subfloat[]{
 \includegraphics[width=0.16\textwidth]{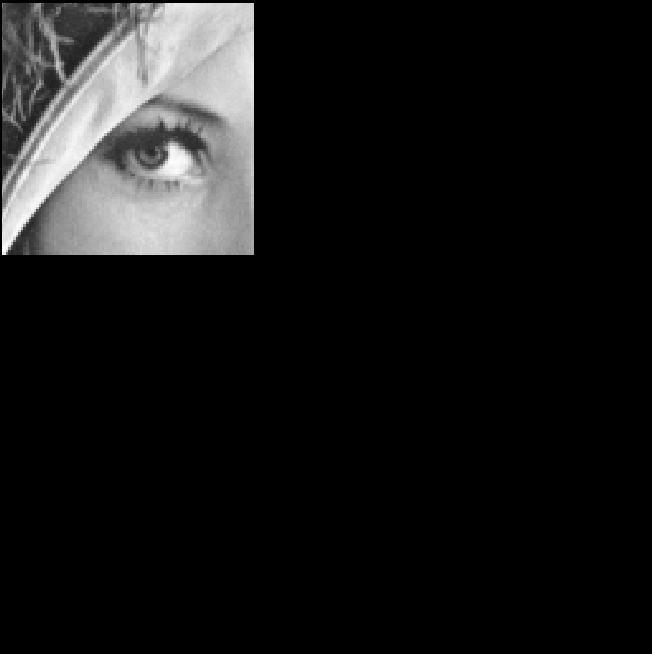}}
 \caption{The results of the De-Ghost method of \secTag~\ref{sec::FiniteGhosts} applied to discrete reconstruction. A total of 101 rational angle projections of a $100\times 100$ image of Lena was exactly reconstructed within a \ac{DFT} space of size $257\times 257$. (a) shows the resulting \ac{DRT} space. (b) shows the resulting Ghosts when (a) is reconstructed. (c) shows a cropped version of (b) focussing on the embedded image. (d) and (e) show the Ghost eigenvalues and the Ghost that was used to recover the reconstructed image exactly (shown in (f)).}
\label{fig::Results}
\end{figure}
The computation time was approximately 500 milliseconds on a 1.6 GHz AMD\texttrademark~Turion\texttrademark~64-bit Laptop. This fairs a few orders of magnitude better than the method of Chandra~\emph{et al.}~\citep{Chandra2008}, whose computations took in the order of hours to complete.

Further work needs to be done to handle inconsistencies within projections while utilising De-Ghosting. Also, a study of the small variations within the rational angle multiplicity may be important in understanding whether some projections are more important than others.

\section*{Conclusion}
A new discrete theory of Ghosts was constructed that allowed one to describe the exact effects of missing periodic slices in the \acf{DFT}. The theory was used to construct new Ghost convolution and de-convolution methods, which can be utilised in recovering missing slices in the \ac{DFT} (see \figsTag~\ref{fig::Kernel} to~\ref{fig::SubCirc} and \propsTag~\ref{thm::Ghosts} to~\ref{thm::KernelProjections}). The methods required the use of the \acl{NTT} and the \acl{NRT} of Chandra~\citep{Chandra2010c} to avoid numerical overflow and other loss of precision problems. This De-Ghost method was then used to solve the discrete inverse problem of reconstructing from limited rational angle projections (see \figsTag~\ref{fig::GhostSet} and~\ref{fig::Results}). This solution has a computational complexity of $O(n\log_2 n)$ with $n = N^2$ and results in an exact reconstruction. A brief study of the limited angle set was also made (see Graphs in \figTag~\ref{fig::Multiplicity}).

\section*{Acknowledgements}
S. Chandra would like to thank the Faculty of Science, Monash University for a Ph.D scholarship and a publications award. N. Normand would like to thank the Australian Research Council for his International Fellowship (ARC LX 0989907).

\appendix
\section{Fast Prime-Length Number Theoretic Transforms}\label{sec::FastNTT}
For a prime-length $\numberSymbol$, the \ac{NTT} can be computed by first selecting the modulus $\modulusSymbol$ as $\modulusSymbol = k\cdot \numberSymbol + 1$. This allows the power of $\primRootSymbol$ in \eqnTag~\eqref{eqn::Euler} to be a multiple of $\numberSymbol$. For example, the modulus for the prime-length $\numberSymbol=101$ is 607 with $k=6$. Then Rader's~\citep{Rader1968} can be used to compute its fast form.

Rader~\citep{Rader1968} constructed a Fast \ac{DFT} algorithm for prime-length \acp{DFT} by turning them into cyclic convolutions of composite lengths. This is achieved by permuting the input vector $\inputVecSymbol$ so that the Transform matrix $\FourierSymbol$ becomes a circulant matrix $\CircVandermondeSymbol$. The algorithm is computed by handling the DC term separately and then applying the following four steps:
\begin{enumerate}
 \item Compute the primitive root $\primRootSymbol$ for the length $N=\primeSymbol$.
 \item Determine the new order $\inputVecSymbol(\firstIndex) \to \inputVecSymbol(\firstIndex')$ of the elements of the input vector $\inputVecSymbol$, where $\firstIndex'$ is given by $\firstIndex' \equiv \primRootSymbol^{-\firstIndex} \imod \primeSymbol \equiv \primRootSymbol^{\primeSymbol-\firstIndex} \imod \primeSymbol$.
 \item Determine the unique row of $\CircVandermondeSymbol$, the circulant form of $\VandermondeSymbol$. This row is given by $\CircVandermondeSymbol_\firstIndex \equiv \rootSymbol^{\primRootSymbol^{\firstIndex}} \imod \primeSymbol$.
 \item If $\primeSymbol-1$ is not highly composite, determine a power of two $\numberSymbol'>2\cdot\primeSymbol-4$. Create a new sequence $\inputVecSymbol'$ of this length and pad it with $\numberSymbol'-\primeSymbol+1$ zeros between the first and second elements of the vector. Create $\CircVandermondeSymbol'$ of length $\numberSymbol'$ by repeating the $\primeSymbol-1$ sequence of $\CircVandermondeSymbol$.
\end{enumerate}
An example of the circulant form $\CircVandermondeSymbol$, as a $4\times 4$ sub-matrix, is given in \eqnTag~\eqref{eqn::Rader} for $\primeSymbol=5$ and $\primRootSymbol=2$. 
\begin{eqnarray}
 \left[
 \begin{array}{ccccc}
 0 & 0 & 0 & 0 & 0\\
 0 & \rootSymbol^1 & \rootSymbol^2 & \rootSymbol^3 & \rootSymbol^4\\
 0 & \rootSymbol^2 & \rootSymbol^4 & \rootSymbol^1 & \rootSymbol^3\\
 0 & \rootSymbol^3 & \rootSymbol^1 & \rootSymbol^4 & \rootSymbol^2\\
 0 & \rootSymbol^4 & \rootSymbol^3 & \rootSymbol^2 & \rootSymbol^1
 \end{array}
 \right] \overset{\textrm{Rader's Reorder}}{\Longrightarrow}
 \left[
 \begin{array}{ccccc}
 0 & 0 & 0 & 0 & 0\\
 0 & \rootSymbol^1 & \rootSymbol^2 & \rootSymbol^4 & \rootSymbol^3\\
 0 & \rootSymbol^3 & \rootSymbol^1 & \rootSymbol^2 & \rootSymbol^4\\
 0 & \rootSymbol^4 & \rootSymbol^3 & \rootSymbol^1 & \rootSymbol^2\\
 0 & \rootSymbol^2 & \rootSymbol^4 & \rootSymbol^3 & \rootSymbol^1
 \end{array}
 \right].\label{eqn::Rader}
\end{eqnarray}
The zeros in the above matrices are present because the DC term is computed separately and the first element of the input vector is added to the result after the convolution. The rest of the transform is completed as the following:
\begin{enumerate}
 \item Determine the eigenvalues of (or diagonalise) $\inputVecSymbol'$ and $\CircVandermondeSymbol'$ by computing the \ac{NTT} of each using the dyadic divide and conquer algorithms~\citep{Cooley1965, Agarwal1974}. 
 \item Multiply the eigenvalues to determine the eigenvalues of the convolution.
 \item Inverse \ac{NTT} the convolution eigenvalues to determine the convolution.
 \item Normalise the convolution by $\primeSymbol-1$ using the multiplicative inverse so that $(\primeSymbol-1)\cdot(\primeSymbol-1)^{-1} \equiv 1 \imod \modulusSymbol$.
 \item Add the first element of $\inputVecSymbol$ to each coefficient of the transform result.
 \item Compute the DC term by summing all terms of the input vector and placing this as the $\FFTfirstIndex=0$ coefficient of the transformed result $\inputVecSymbol(\FFTfirstIndex)$.
 \item Extract the rest of the coefficients of $\inputVecSymbol(\FFTfirstIndex)$ for $0 < \FFTfirstIndex \leqslant \primeSymbol-1$ of the prime-length transform elements by using $\FFTfirstIndex \equiv \primRootSymbol^{\firstIndex} \imod \primeSymbol$ from the normalised convolution.
 \item Normalise the result by $\primeSymbol$ if desired using multiplicative inverse so that $\primeSymbol\cdot\primeSymbol^{-1} \equiv 1 \imod {\primeSymbol'}$.
\end{enumerate}
The implementation of this transform can be found in the \ac{FTL} library~\citep{Chandra2009b}.

\footnotesize
\bibliographystyle{elsarticle-num}
\bibliography{RadonJabRef}

\acrodef{1D}{one dimensional}
\acrodef{2D}{two dimensional}
\acrodef{RT}{Radon Transform}
\acrodef{CT}{Computed Tomography}
\acrodef{DT}{Discrete Tomography}
\acrodef{FT}{Fourier Transform}
\acrodef{FST}{Fourier Slice Theorem}
\acrodef{DRT}{Discrete Radon Transform}
\acrodef{NRT}{Number-Theoretic Radon Transform}
\acrodef{FRT}{Fast Radon Transform}
\acrodef{MT}{Mojette Transform}
\acrodef{FMT}{Fast Mojette Transform}
\acrodef{FFT}{Fast Fourier Transform}
\acrodef{DFT}{Discrete Fourier Transform}
\acrodef{NTT}{Number Theoretic Transform}
\acrodef{DPM}{Dirac Pixel Model}
\acrodef{2PSE}{2-Point Structuring Element}
\acrodef{CBP}{Circulant Back-Projection}
\acrodef{PCT}{Projection Convolution Theorem}
\acrodef{FTL}{Finite Transform Library}

\end{document}